\pdfoutput=1
\documentclass[onecolumn,12pt]{IEEEtran}

\usepackage{amsmath,amssymb,amsthm,mathtools,dsfont}
\usepackage{graphicx}
\usepackage{booktabs}
\usepackage[colorlinks=true,allcolors=blue]{hyperref}
\usepackage[capitalise]{cleveref}

\theoremstyle{plain}
\newtheorem{theorem}{Theorem}
\newtheorem{lemma}[theorem]{Lemma}
\newtheorem{corollary}[theorem]{Corollary}
\newtheorem{proposition}[theorem]{Proposition}
\theoremstyle{definition}
\newtheorem{definition}[theorem]{Definition}
\newtheorem{remark}[theorem]{Remark}

\DeclareMathOperator*{\esssup}{ess\,sup}

\newcommand{\cC}{{\mathcal C}}
\newcommand{\cD}{{\mathcal D}}\newcommand{\cE}{{\mathcal E}}

\newcommand{\cJ}{{\mathcal J}}
\newcommand{\cP}{{\mathcal P}}

\newcommand{\cU}{{\mathcal U}}\newcommand{\cV}{{\mathcal V}}
\newcommand{\cX}{{\mathcal X}}\newcommand{\cY}{{\mathcal Y}}

\newcommand{\mR}{{\mathbb R}}

\newcommand{\BRA}[1]{\left( #1 \right)}
\newcommand{\BRAb}[1]{\left[ #1 \right]}

\newcommand{\BRAs}[1]{\left\{ #1 \right\}}

\newcommand{\PR}[1]{\mathbb{P}\left\{ #1 \right\}}
\newcommand{\PRs}[2]{\mathbb{P}_{#1}\left\{ #2 \right\}}
\newcommand{\E}[1]{\mathbb{E}\left( #1 \right)}
\newcommand{\Es}[2]{\mathbb{E}_{#1}\left( #2 \right)}

\newcommand{\Unif}[1]{\cU\left( #1 \right)}
\newcommand{\uU}{\Unif{[0,1]}}
\newcommand{\Ind}[1]{\mathds{1}_{\BRAs{#1}}}

\newcommand{\W}{W_{Y|X}}

\newcommand{\BETA}[3]{\beta_{#1}\!\Big( #2\; ; \; #3\Big)}
\newcommand{\PEP}[2]{p_e\BRA{#1 \mid #2}}
\newcommand{\PEPs}[1]{p_e\BRA{#1}}
\newcommand{\PEPU}[3]{p_e\BRA{#1 \mid #2; #3}}
\newcommand{\PEC}[2]{p_c\BRA{#1 \mid #2}}
\newcommand{\PECU}[3]{p_c\BRA{#1 \mid #2; #3}}
\newcommand{\Fspec}[2]{F\!\BRA{#1\,;\,#2}}
\newcommand{\Dspec}[2]{D_s\!\BRA{#1\,,\,#2}}
\newcommand{\Fspecjscc}[1]{F_V\!\BRA{#1}}

\newcommand{\RCUp}{\mathrm{RCU}^{+}}
\newcommand{\verdu}{Verd\'{u}}

\newcommand{\etal}{\textit{et al.}}

\begin{document}

\title{One-Shot Information Theory via the Pairwise Error Probability:\\
Lossy, Joint Source--Channel, Erasure, and Multiuser Coding}

\author{Nir~Elkayam~and~Meir~Feder%
\thanks{The authors are with the Department of Electrical Engineering--Systems,
Tel Aviv University, Tel Aviv, Israel (e-mail: nir.elkayam@gmail.com; meir@eng.tau.ac.il).}%
\thanks{This paper is a companion to~\cite{elkayampep1}, which develops the
pairwise-error-probability framework for point-to-point channel coding.}}

\maketitle

\begin{abstract}
This paper extends a one-shot (finite-blocklength)
information-theoretic framework built on a single primitive: the \emph{pairwise error probability}
(PEP) of a randomized, dither-broken decoding rule, and the \emph{error spectrum}
it induces. A companion paper developed the framework for
point-to-point channel coding --- uniformity of the PEP, the error-spectrum
representation of achievability and converse, and the linear-programming form of
the prior-optimized minimax meta-converse. Here we show that the same primitive,
read on an enlarged candidate space, governs four further settings: lossy source
coding under average distortion, joint source--channel coding with list decoding,
channel coding with an erasure/undetected-error option, and the two-user
multiple-access channel. In each case a single spectrum yields a random-coding
achievability bound and exact fixed-code identities, and we indicate how the
convex --- indeed linear-programming --- prior optimization of the
channel-coding case extends under matched decoding. The development recovers the one-shot lossy bound of
Matsuta--Uyematsu and the joint source--channel bounds in the
Csisz\'ar tradition, complements the lossy bounds of Kostina--\verdu{}, and connects the multiuser case to the comparable
achievability/converse pair through three pairwise error events.
\end{abstract}


\section{Introduction}\label{sec:intro}

A companion paper~\cite{elkayampep1} built one-shot channel coding on a single
primitive: the dithered \emph{pairwise error probability} (PEP) and the
\emph{error spectrum} it induces. It showed that one functional, read through
different kernels, delivers the random-coding achievability bound, the exact
fixed-code converse, and (under matched decoding) the prior-optimized minimax
meta-converse as a polynomial-size linear program. The mechanism behind every
such kernel is the lift: for a random codebook the competitors are i.i.d., so
conditional on the transmitted pair the codebook error probability is exactly
$1-(1-p_e)^{M-1}$ --- a deterministic function of the pairwise quantity --- and
every bound is this lift or a variant, that is, a kernel applied to the PEP.

The construction never used
that the candidate set was a channel codebook: the PEP is defined for an
arbitrary candidate space and an arbitrary decoding metric. This paper exploits
that latitude. We instantiate the same primitive on four enlarged candidate
spaces and obtain, in each, the same spectrum-governed development (for the
erasure option, achievability and a correct-decoding converse); in the
single-terminal settings the convex prior optimization of~\cite{elkayampep1}
extends as well, and we indicate how, without redeveloping it:
\begin{itemize}
\item \textbf{Lossy source coding} (\cref{sec:rd}): the lift runs in the
\emph{value} domain --- the pairwise \emph{correct} probability is the rank
transform of the distortion --- and the same construction yields exact
random-coding distortion and an exact fixed-code converse, recovering the
one-shot bound of Matsuta--Uyematsu~\cite{matsuta2015non} and complementing
the $d$-tilted-information bound of Kostina--\verdu~\cite{kostina2012fixed}.
\item \textbf{Joint source--channel coding with list decoding}
(\cref{sec:jscc}): enlarging the candidate space to source--codeword pairs gives
a joint spectrum in which the list size enters as a rate offset,
with fixed-code exactness extending V\'azquez-Vilar's matched-code
identity~\cite{vazquez2016bayesian}.
\item \textbf{Erasure/undetected-error decoding} (\cref{sec:erasure}): a
two-threshold decision region over the PEP produces achievability bounds for
the erasure-option problem, and the high-rate correct-decoding regime with its
Arimoto-type strong converse falls out of the complementary spectrum.
\item \textbf{The two-user multiple-access channel} (\cref{sec:network}): three
pairwise error events (joint, message-1, message-2) give a sum-spectrum
achievability bound and a max-spectrum converse, preserving the comparable
achievability/converse structure into the multiuser setting.
\end{itemize}
Throughout, the framework results of~\cite{elkayampep1} are used as black boxes; the next
section recaps the few we need, so the present paper is readable on its own. The
emphasis here is on the \emph{reach} of the primitive --- that one construction,
unchanged, covers source, joint, erasure, and multiuser coding --- rather than on
new converse machinery.

\section{Recap of the PEP Framework}\label{sec:pep}

We recall, without proof, the results of~\cite{elkayampep1} used below. The
guiding observation is that only the \emph{relative order} of metric scores
matters: any monotone transformation of the metric gives the same decoder, and
the PEP is the canonical representative of this order-equivalence class ---
uniquely, its value is intrinsic, being the probability of being outranked. At
fixed $y$ it is uniform on $[0,1]$; its law under the channel-joint
input--output law --- the error spectrum --- is the non-trivial object.

Fix finite alphabets $\cX,\cY$, a prior $Q_X\in\cP(\cX)$, and a measurable metric
$m:\cX\times\cY\to[0,\infty]$. Attaching an i.i.d.\ dither $U\sim\uU$ to each
candidate and comparing score--dither pairs $(m(x,y),u)$ in the weak
lexicographic order (a higher score wins; at equal scores, a smaller dither
wins), the \emph{dithered PEP} of $(x,y,u)$ is
$\PEPU{x}{y}{u}\triangleq G_{x,y}+u\,H_{x,y}$, where $G_{x,y}=Q_X\{m(\cdot,y)>m(x,y)\}$
and $H_{x,y}=Q_X\{m(\cdot,y)=m(x,y)\}$, and $\PEP{x}{y}\triangleq\PEPU{x}{y}{U}$.

\begin{proposition}[PEP properties]\label{thm:pep-uniform}
Let $(X,U)\sim Q_X\times\uU$ and fix $y\in\cY$. Then:
\begin{enumerate}
\item[(i)]\emph{(Lexicographic equivalence.)} For every fixed $(x,y,u)$,
\[
\PRs{X,U}{(m(X,y),U)\succ(m(x,y),u)}=\PRs{X,U}{\PEP{X}{y}<\PEPU{x}{y}{u}}=\PEPU{x}{y}{u}.
\]
\item[(ii)]\emph{(Order preservation.)} If $m(x_1,y)>m(x_2,y)$ then
$\PEPU{x_1}{y}{u_1}\le\PEPU{x_2}{y}{u_2}$ for all $u_1,u_2\in[0,1]$.
\item[(iii)]\emph{(Uniformity.)} $\PEP{X}{y}\sim\uU$; this lifts to
$\PEP{X}{Y}\sim\uU$ under the product law $(X,Y)\sim Q_X\times Q_Y$ for any
output measure $Q_Y$.
\item[(iv)]\emph{(Tail--quantile inverse map.)} There is a measurable map
$\tilde m(\cdot,y)$, non-increasing in its first argument, with
$\tilde m(\PEP{X}{y},y)=m(X,y)$ a.s.
\end{enumerate}
\end{proposition}

Parts (i)--(iii) are proved in~\cite{elkayampep1}. Part (iv), which
\cite{elkayampep1} does not use, is proved in Appendix~\ref{app:rank-recall} from the
randomized probability-integral transform (Lemma~\ref{app:lem-rpit}).

The \emph{error spectrum} is the CDF of the negative log-PEP under the
channel-induced joint law,
\begin{equation}\label{eq:def-spectrum}
\Fspec{Q_X}{z}\triangleq\PR{-\log\PEP{X}{Y}\le z},\qquad(X,Y)\sim Q_X\cdot\W,
\end{equation}
which (unlike the product-law PEP) is non-trivial and carries the dependence on
$Q_X$ and the channel.

\begin{proposition}[Atomlessness~{\cite{elkayampep1}}]\label{lem:atomless}
Let $Q_X^{\mathrm{pep}}$ define the PEP and let $Q_X^{\mathrm{prob}}\ll
Q_X^{\mathrm{pep}}$; assume moreover $Q_X^{\mathrm{prob}}\cdot\W\ll
Q_X^{\mathrm{prob}}\times P_Y$ for the induced output marginal $P_Y$ ---
automatic for discrete alphabets. Then under
$(X,Y)\sim Q_X^{\mathrm{prob}}\cdot\W$ the PEP has no atoms; consequently
$z\mapsto\Fspec{Q_X}{z}$ is continuous.
\end{proposition}

\subsection{Channel coding}\label{sec:channel}

For $M=\lceil e^{R}\rceil$ codewords drawn i.i.d.\ from $Q_X$, the ensemble
error of the dithered maximum-metric decoder is exact~\cite{elkayampep1}:
conditional on the transmitted pair the competitors are i.i.d., so
\begin{equation}\label{eq:recap-exact-rc}
\bar P_e(R;Q_X)=\Es{X,Y,U}{1-\bigl(1-\PEPU{X}{Y}{U}\bigr)^{M-1}}.
\end{equation}
The refined-union ($\RCUp$) relaxation of \eqref{eq:recap-exact-rc} and the
exact fixed-code converse of~\cite{elkayampep1} read
\begin{equation}\label{eq:recap-rcu}
\bar P_e(R;Q_X)\le\Es{X,Y,U}{\min\BRA{1,e^{R}\PEPU{X}{Y}{U}}}
=e^{R}\!\int_R^\infty e^{-z}\Fspec{Q_X}{z}\,dz,
\end{equation}
\begin{equation}\label{eq:recap-converse}
P_e(\cC)=\PR{\PEP{X}{Y}\ge 1/M}=\Fspec{Q_X^\cC}{R}\ \ \text{(empirical prior $Q_X^\cC$; converse-natural $M=e^{R}$, \cref{rem:size-conv}),}
\end{equation}
so achievability and converse are kernel/threshold readings of the same spectrum.
The identity \eqref{eq:recap-converse} is that of the \emph{dithered}
maximum-metric decoder; in the matched case it is tie-rule-independent and
equals the optimal (MAP) error probability~\cite{elkayampep1}. Writing
$\tilde P_e(R;Q_X)$ for the right-hand side of \eqref{eq:recap-rcu} and
$E(R)\triangleq-\log\tilde P_e(R;Q_X)$, the slope identity of~\cite{elkayampep1} reads
\begin{equation}\label{eq:recap-slope}
\tilde P_e(R;Q_X)=\frac{\Fspec{Q_X}{R}}{1+\dot E(R)},\qquad -1<\dot E(R)\le0,
\end{equation}
at every $R$ with $\Fspec{Q_X}{R}>0$; the general list-size ($L$-ary) version is
Lemma~\ref{app:lem-integral-L} in Appendix~\ref{app:jscc}.

\begin{remark}[Codebook-size conventions]\label{rem:size-conv}
Three size--rate conventions appear in this paper, each where it is natural:
the default $M=\lceil e^{R}\rceil$ above; the \emph{converse-natural} $M=e^{R}$,
under which a size-$M$ code has rate exactly $R=\log M$; and the
\emph{achievability-natural} $M-1=e^{R}$, which counts competitors rather than
codewords and is used in the correct-decoding analysis of \cref{subsec:era-cd}.
The three differ by at most one codeword; later uses point back here rather
than re-deriving the bookkeeping.
\end{remark}

\subsection{Identities and prior optimization}\label{sec:prioropt}

Three identities make the spectrum computable. Writing
$\BETA{\alpha}{P}{Q}$ for the Neyman--Pearson functional and
$((Q_X\times P_Y)\cdot m)(x,y)=Q_X(x)P_Y(y)m(x,y)$ for the metric-tilted measure:

\begin{proposition}[Meta-converse identity~{\cite{elkayampep1}}]\label{thm:meta-converse}
With $Z_e=-\log\PEP{X}{Y}$, $Z_c=-\log\PEC{X}{Y}$ under $(X,Y)\sim Q_X\times P_Y$,
and assumptions \textup{(A1)} $Q_X\{m(\cdot,y)<\infty\}=1$, \textup{(A2)}
$\Es{Q_X\times P_Y}{m}<\infty$,
\begin{align}
\Es{Q_X\times P_Y}{m\,\Ind{Z_e\le R}}
&=\max_{Q_Y}\BETA{1-e^{-R}}{Q_X\times Q_Y}{(Q_X\times P_Y)\cdot m},\label{eq:meta-converse-e}\\
\Es{Q_X\times P_Y}{m\,\Ind{Z_c\ge R}}
&=\max_{Q_Y}\BETA{e^{-R}}{Q_X\times Q_Y}{(Q_X\times P_Y)\cdot m}.\label{eq:meta-converse-c}
\end{align}
In the matched case $\Fspec{Q_X}{R}=\max_{Q_Y}\BETA{1-e^{-R}}{Q_X\times Q_Y}{Q_X\W}$.
\end{proposition}

The error-tail identity \eqref{eq:meta-converse-e} is proved
in~\cite{elkayampep1}; the complementary-tail twin
\eqref{eq:meta-converse-c}, noted there without proof, is proved
in Appendix~\ref{app:rank-recall}.

\begin{proposition}[Reverse-channel identity~{\cite{elkayampep1}}]\label{thm:reverse-channel}
For $z>0$, with $D_\infty(W\|Q_X\mid Q_Y)\le z$ the pointwise cap
$W(x\mid y)\le e^{z}Q_X(x)$,
\begin{align}
\Es{Q_X\times Q_Y}{m\,\Ind{\PEC{X}{Y}\le e^{-z}}}
&=\inf_{D_\infty(W\|Q_X\mid Q_Y)\le z}e^{-z}\Es{W\cdot Q_Y}{m},\label{eq:rc-inf}\\
\Es{Q_X\times Q_Y}{m\,\Ind{\PEP{X}{Y}\le e^{-z}}}
&=\sup_{D_\infty(W\|Q_X\mid Q_Y)\le z}e^{-z}\Es{W\cdot Q_Y}{m}.\label{eq:rc-sup}
\end{align}
\end{proposition}

\begin{proposition}[Joint convexity~{\cite{elkayampep1}}]\label{thm:joint-convex}
In the matched case $Q_X\mapsto\Fspec{Q_X}{z}$ is convex for every $z$;
consequently the prior-optimized minimax meta-converse $\inf_{Q_X}\Fspec{Q_X}{R}$
is the value of a polynomial-size linear program in $(Q_X,W^*_{X|Y})$.
\end{proposition}

We also use the closed-form fact that for Gallager-symmetric channels the
uniform prior is optimal, in both coding directions and at every
blocklength~\cite{elkayampep1}.

The four extensions below reuse Propositions~\ref{thm:pep-uniform}--\ref{thm:joint-convex}
verbatim, changing only the candidate space and the metric. In each setting we
indicate how the prior-optimization linear program of~\cite{elkayampep1}
extends; we do not redevelop it here.

\section{Extension: Lossy Source Coding}\label{sec:rd}

In rate--distortion the quantity of interest is a \emph{value} --- the best distortion among $M$ reproductions --- not an event, so the rank-only lift of channel coding is not enough. The mechanism is to run the same lift in the value domain: the pairwise correct probability defined below is the rank transform \emph{of the distortion} --- uniform on $[0,1]$ for each source symbol $x$, with a non-decreasing inverse map recovering $d(x,Y)$ from it almost surely. Monotonicity passes the minimum through, $\min_i d(x,Y_i)=\tilde d\bigl(x,\min_i U_i\bigr)$, and the minimum of $M$ i.i.d.\ uniforms has a closed-form law; mapping back yields the value-weighted distortion spectrum that governs both coding directions below.

Concretely, the PEP framework transfers to lossy source coding by relocating the primitive to the \emph{output} side: a reproduction point either does or does not beat a candidate in distortion, and the same randomized rank transform that uniformized the PEP now uniformizes a pairwise \emph{correct} probability. The development parallels Section~\ref{sec:channel} term for term, with the error spectrum replaced by a distortion spectrum and the meta-converse and reverse-channel identities of Section~\ref{sec:pep} reused under a role transposition. We treat fixed-length codes under the average-distortion criterion; the excess-distortion specialization is the subject of \cref{subsec:excess}, and the proofs deferred below are collected in Appendix~\ref{app:rd}. One-shot achievability and converse bounds for this setting, including variational forms of the governing functional, were derived in our earlier work~\cite{elkayam2020one} (full version~\cite{elkayam2019oneshot_rate_distortion_full}); new in the present treatment are the derivation through the pairwise-correct rank transform, the exactness of the fixed-code converse (\cref{thm:rd-converse}), and the reverse-channel/linear-programming route to prior optimization (\cref{rem:rd-lp}).

Fix a source $X\sim P_X$ on a finite alphabet $\cX$, a finite reproduction alphabet $\cY$, and a bounded distortion function $d:\cX\times\cY\to\mR_+$ with $d_{\max}=\sup_{x,y}d(x,y)$. A code is a reproduction codebook $\cC\subseteq\cY$ of size $M$; its distortion and the optimal distortion--rate function are
\begin{equation*}
D(\cC)\triangleq\Es{P_X}{\min_{y\in\cC}d(X,y)},\qquad
D(R)\triangleq\min_{|\cC|=\lceil e^{R}\rceil}D(\cC).
\end{equation*}

\subsection{Pairwise correct probability}

The dual of the PEP records how often a random reproduction strictly improves on a candidate.

\begin{definition}[Pairwise correct probability]\label{def:pec}
For $x\in\cX$, $y\in\cY$, and $u\in[0,1]$, with $Y\sim Q_Y$,
\begin{equation*}
\PECU{y}{x}{u}\triangleq Q_Y\BRAs{d(x,Y)<d(x,y)}+u\,Q_Y\BRAs{d(x,Y)=d(x,y)},
\end{equation*}
and the randomized pairwise correct probability is $\PEC{y}{x}\triangleq\PECU{y}{x}{U}$ with $U\sim\uU$ independent of $Y$.
\end{definition}

This is the PEP construction read on the output alphabet and ordered by ascending distortion: a small value of $\PEC{y}{x}$ marks a good reproduction. (The same symbol $p_c$ appears in \cref{sec:prioropt} as the complement-of-PEP $\PEC{X}{Y}$; the argument-order convention --- candidate given observation there, reproduction given source symbol here --- distinguishes the two.) The dual uniformity and its inverse follow exactly as for the PEP.

\begin{lemma}[Uniformity and inverse distortion map]\label{lem:pec-uniform}
Fix $x\in\cX$, let $Y\sim Q_Y$ and $U\sim\uU$ be independent. Then $\PEC{Y}{x}\sim\uU$. Moreover there is a measurable map $\tilde d:\cX\times[0,1]\to\mR_+$, non-decreasing in its second argument, such that $\tilde d\!\BRA{x,\PEC{Y}{x}}=d(x,Y)$ holds $Q_Y$-a.s.
\end{lemma}

\begin{IEEEproof}
The randomized rank transform of Section~\ref{sec:pep} (\cref{thm:pep-uniform}) is distribution-preserving regardless of the ordering direction: applying the RPIT (Lemma~\ref{app:lem-rpit}) to $T=d(x,Y)$ with $Y\sim Q_Y$, ordered ascending, gives $\PEC{Y}{x}\sim\uU$ together with the non-decreasing inverse quantile $\tilde d(x,\cdot)$ satisfying $\tilde d\!\BRA{x,\PEC{Y}{x}}=d(x,Y)$ $Q_Y$-a.s.
\end{IEEEproof}

The map $\tilde d$ makes $\PEC{Y}{x}$ an invertible ``distortion rank,'' which drives the exact random-coding analysis below.

\subsection{Achievability via random coding}

For a codebook drawn i.i.d.\ from $Q_Y$, write $D_{\mathrm{rand}}(Q_Y,R)\triangleq\Es{\cC\sim Q_Y^{M}}{D(\cC)}$ with $M=\lceil e^{R}\rceil$. The output-side analogue of the error spectrum is the distortion spectrum.

\begin{definition}[Distortion spectrum]\label{def:dspec}
For $z>0$ and $Q_Y\in\cP(\cY)$,
\begin{equation*}
\Dspec{z}{Q_Y}\triangleq e^{z}\,\Es{P_X\times Q_Y}{d(X,Y)\,\Ind{\PEC{Y}{X}\le e^{-z}}}.
\end{equation*}
\end{definition}

By Lemma~\ref{lem:pec-uniform} the event $\{\PEC{Y}{X}\le e^{-z}\}$ has $(Q_Y\times\uU)$-mass exactly $e^{-z}$ for every $x$, so the prefactor $e^{z}$ rescales to a conditional expectation: $\Dspec{z}{Q_Y}$ is the expected distortion of the best-of-$M$ reproduction at $M=e^{z}$ under the uniform-on-codebook prior. It is non-increasing and continuous in $z$ (via the absolutely continuous primitive $A_x$ in the proof of \cref{thm:rd-rc}, Appendix~\ref{app:rd}) and convex in $Q_Y$ (\cref{rem:rd-lp}).

\begin{theorem}[Exact random-coding distortion]\label{thm:rd-rc}
Let $d$ be bounded, $Q_Y\in\cP(\cY)$, $R>0$, $M=\lceil e^{R}\rceil$, and
\begin{equation*}
K_M(z)\triangleq M(M-1)\,e^{-2z}\BRA{1-e^{-z}}^{M-2},\qquad z\ge0,
\end{equation*}
which is a probability density on $[0,\infty)$. Then
\begin{equation*}
D_{\mathrm{rand}}(Q_Y,R)=\int_{0}^{\infty}\Dspec{z}{Q_Y}\,K_M(z)\,dz.
\end{equation*}
\end{theorem}

\begin{IEEEproof}[Proof sketch]
By the inverse map of Lemma~\ref{lem:pec-uniform}, monotonicity passes the minimum through,
\begin{equation*}
\min_i d(x,Y_i)=\tilde d\BRA{x,\min_i\PEC{Y_i}{x}},
\end{equation*}
and $\min_i\PEC{Y_i}{x}$ is the minimum of $M$ i.i.d.\ uniforms, with density $M(1-w)^{M-1}$. Integration by parts against the cumulative distortion primitive, followed by the change of variables $w=e^{-z}$, produces the kernel $K_M$; see Appendix~\ref{app:rd} for the full argument.
\end{IEEEproof}

Localizing the integral by monotonicity of $\Dspec{z}{Q_Y}$ gives, for $0<\gamma<R$,
\begin{equation*}
D_{\mathrm{rand}}(Q_Y,R)\le\Dspec{R-\gamma}{Q_Y}+d_{\max}\,e^{-e^{\gamma}(1-e^{-R})},
\end{equation*}
an explicit achievability bound: a horizontal shift $\gamma$ of the spectrum plus an exponentially small correction.

\subsection{Fixed-code converse}

The rank interpretation makes the spectrum \emph{exact} for any fixed code: the minimum-distortion reproduction is the one of smallest pairwise correct probability.

\begin{theorem}[Fixed-code converse]\label{thm:rd-converse}
For any code $\cC$ of size $M$,
\begin{equation*}
D(\cC)=\Dspec{\log M}{Q_Y^{\cC}},
\end{equation*}
where $Q_Y^{\cC}$ is the uniform distribution on $\cC$. Consequently, for the converse-natural size $M=e^{R}$ (\cref{rem:size-conv}), $D(\cC)=\Dspec{R}{Q_Y^{\cC}}\ge\inf_{Q_Y}\Dspec{R}{Q_Y}$.
\end{theorem}

\begin{IEEEproof}[Proof sketch]
Apply the finite rank identity (Lemma~\ref{app:lem-rank}), ordering distortions ascending: the dither-expected weight $\Es{U}{\Ind{\PEC{j}{x}\le1/M}}$ is supported on and uniform over the minimal-distortion codewords, so
\begin{equation*}
\min_i d(x,y_i)=M\,\Es{Q_Y^{\cC}}{d(x,Y)\Ind{\PEC{Y}{x}\le1/M}}.
\end{equation*}
Averaging over $X$ gives $\Dspec{\log M}{Q_Y^{\cC}}$; see Appendix~\ref{app:rd} for the full argument.
\end{IEEEproof}

Theorems~\ref{thm:rd-rc} and~\ref{thm:rd-converse} sandwich the optimal distortion through one functional evaluated at shifted rates:
\begin{equation*}
\inf_{Q_Y}\Dspec{R}{Q_Y}\;\le\;D(R)\;\le\;\inf_{Q_Y}\BRAs{\Dspec{R-\gamma}{Q_Y}+d_{\max}\,e^{-e^{\gamma}(1-e^{-R})}}.
\end{equation*}

\subsection{Reverse-channel variational form}

The output-side primitive inherits the two identities of Section~\ref{sec:pep} under the role transposition that sends the observation to the source $X\sim P_X$, the codeword to the candidate $Y\sim Q_Y$, and the PEP to $\PEC{Y}{X}$.

\begin{theorem}[Variational forms of $\Dspec{z}{Q_Y}$]\label{thm:rd-variational}
For $z>0$,
\begin{align}
\Dspec{z}{Q_Y}&=e^{z}\sup_{Q_X}\BETA{e^{-z}}{Q_X\times Q_Y}{P_X\times Q_Y\cdot d},\label{eq:rd-beta}\\
\Dspec{z}{Q_Y}&=\inf_{\substack{\W:\;D_\infty(P_X\W\,\|\,P_X\times Q_Y)\le z}}\Es{P_X\W}{d(X,Y)},\label{eq:rd-dinf}
\end{align}
where $P_X\times Q_Y\cdot d$ denotes the measure with density $d(x,y)$ relative to $P_X\times Q_Y$, the supremum is over auxiliary $Q_X\in\cP(\cX)$, the infimum is over test channels $\W:\cX\to\cY$, and $D_\infty$ is the R\'enyi-$\infty$ divergence.
\end{theorem}

\begin{IEEEproof}[Proof sketch]
Both are the meta-converse (\cref{thm:meta-converse}) and reverse-channel (\cref{thm:reverse-channel}) identities of Section~\ref{sec:pep} under the role transposition above; the shared $P_X$ marginal cancels in the density ratio, $\tfrac{P_X(x)\W(y|x)}{P_X(x)Q_Y(y)}=\tfrac{\W(y|x)}{Q_Y(y)}$, identifying the conditional and joint $D_\infty$; see Appendix~\ref{app:rd}.
\end{IEEEproof}

\begin{remark}[Prior optimization]\label{rem:rd-lp}
By the same role transposition of the joint-convexity identity
(\cref{thm:joint-convex}), $\Dspec{z}{Q_Y}$ is convex in $Q_Y$, so the prior
optimization $\inf_{Q_Y}\Dspec{z}{Q_Y}$ is a convex program. Note that $Q_Y$
enters $\Dspec{z}{Q_Y}$ in two roles at once --- as the competitor prior inside
$\PEC{\cdot}{\cdot}$ and as the sampling law of $Y$ --- exactly as $Q_X$ enters
the channel-coding spectrum \eqref{eq:def-spectrum} in both roles; the
role-transposition argument of Appendix~\ref{app:rd} carries the two roles
simultaneously. It is computable by
exactly the route of~\cite{elkayampep1}: the reverse-channel
form~\eqref{eq:rd-dinf} linearizes the problem, with the test channel $\W$
playing the role of the reverse channel of~\cite{elkayampep1}, the $D_\infty$ cap
the same pointwise box constraint, and the roles of input prior and reproduction
prior transposed. The resulting linear program, its polynomial-size type-based
blocklength reduction, and the companion prior optimization of the achievability
side all follow the development of~\cite{elkayampep1} verbatim; we do not repeat
them here.
\end{remark}

This places the present results inside the finite-blocklength rate-distortion literature (see also the fixed-blocklength converse of Palzer and Timo~\cite{palzer2016converse}). After optimizing $Q_Y$, the excess-distortion specialization of~\eqref{eq:rd-dinf} reproduces the Matsuta--Uyematsu one-shot bound \cite{matsuta2015non}, since
\begin{equation*}
\inf_{Q_Y}D_\infty(P_X\W\,\|\,P_X\times Q_Y)=\log\sum_y\sup_x\W(y|x)
\end{equation*}
(here and below $\sup_x$ ranges over $\operatorname{supp}(P_X)$), and complements the Kostina--Verd\'u $d$-tilted-information bound \cite{kostina2012fixed}. The difference is one of order of optimization: Matsuta--Uyematsu fix a test channel and optimize the prior, whereas the PEP route starts from the prior $Q_Y$ and lets the test channel arise from~\eqref{eq:rd-dinf}; the latter order is what exposes the tractable prior optimization.

\subsection{Excess distortion}\label{subsec:excess}

Replacing the average criterion by an excess-distortion threshold $d_{\mathrm{th}}$ specializes the spectrum to a closed form. Put $d_e(x,y)\triangleq\Ind{d(x,y)>d_{\mathrm{th}}}$ and let $q(x)\triangleq Q_Y\BRAs{y:d(x,y)\le d_{\mathrm{th}}}$ be the $Q_Y$-mass of the $d_{\mathrm{th}}$-ball at $x$. The corresponding distortion spectrum $D_e(z,Q_Y)\triangleq e^{z}\Es{P_X\times Q_Y}{d_e(X,Y)\Ind{\PEC{Y}{X}\le e^{-z}}}$ admits an explicit evaluation.

\begin{lemma}[Closed form of $D_e$]\label{lem:excess}
For the indicator distortion,
\begin{equation*}
D_e(z,Q_Y)=e^{z}\,\Es{P_X}{\BRA{e^{-z}-q(X)}_+}=e^{z}\!\int_0^{e^{-z}}\!\PR{q(X)\le t}\,dt.
\end{equation*}
\end{lemma}

\begin{IEEEproof}[Proof sketch]
Conditioning on $X=x$ with $q=q(x)$: in the $d_e=1$ branch (probability $1-q$) one has $\PEC{Y}{x}=q+U(1-q)$, so $\PR{d_e=1,\,\PEC{Y}{x}\le w\mid X=x}=(w-q)_+$. Averaging over $X$ gives the first form; $(w-q)_+=\int_0^w\Ind{q\le t}\,dt$ gives the second. See Appendix~\ref{app:rd}.
\end{IEEEproof}

Optimizing the reproduction prior recovers, exactly, the Matsuta--Uyematsu one-shot rate--distortion function~\cite{matsuta2015non}: with $R_e(\varepsilon;Q_Y)\triangleq\inf_{\W:\PRs{P_X\W}{d>d_{\mathrm{th}}}\le\varepsilon}D_\infty\BRA{P_X\W\,\|\,P_X\times Q_Y}$,
\begin{equation*}
\inf_{Q_Y}R_e(\varepsilon;Q_Y)=\inf_{\W:\PRs{P_X\W}{d>d_{\mathrm{th}}}\le\varepsilon}\log\sum_y\sup_x\W(y|x)=R_{\mathrm{MU}}(\varepsilon),
\end{equation*}
the middle equality being $\inf_{Q_Y}D_\infty(P_X\W\|P_X\times Q_Y)=\log\sum_y\sup_x\W(y|x)$ \cite[Lem.~4]{matsuta2015non}. The two formulations differ only in the order of optimization noted after \cref{rem:rd-lp}; fixing the prior $Q_Y$ first (e.g.\ the code-induced uniform law) is what renders the prior optimization a tractable LP.

\section{Extension: Joint Source--Channel Coding with List Decoding}\label{sec:jscc}

The PEP primitive of Sections~\ref{sec:pep}--\ref{sec:channel} is not tied to channel coding. The same uniformity property and meta-converse identity apply once the candidate space is enlarged from $\cX$ to $\cV\times\cX$: a source--codeword pair $(v,x)$ plays the role of the codeword, and the dithered PEP of Section~\ref{sec:pep} is well defined on this enlarged space.

Two generalizations fall out of the one construction. \emph{List decoding} lets the decoder output $L$ candidates; the list size enters as a uniform rate offset $\log L$ on both achievability and converse. \emph{Non-uniform sources} let information symbols follow a prescribed $P_V$; this is absorbed into the metric. One-shot list-decoding bounds for channel coding in this spirit appear in our earlier note~\cite{ElkayamLD2017}; new in the present treatment are the joint source--channel candidate-space lift and the fixed-code exactness of \cref{thm:jscc-converse}. Everything below is stated with a proof idea in the body; the full proofs are collected in Appendix~\ref{app:jscc}.

\subsection{Setup}

Fix a finite source alphabet $\cV$, a source distribution $P_V$ on $\cV$, a channel $\W:\cX\to\cY$, and a decoding metric $m:\cV\times\cX\times\cY\to[0,\infty]$. A codebook $\cC=\BRAs{x_v:v\in\cV}$ assigns a channel input $x_v$ to each source symbol; in the random-coding ensemble $X\sim Q_{X|V}(\cdot\mid v)$ given $V=v$. Given $y$, the list decoder outputs the set $\cD(y)\subseteq\cV$ of the $L$ symbols with the largest scores $m(v,x_v,y)$, ties broken by independent dithers as in Section~\ref{sec:channel}. An error occurs when the transmitted symbol is not listed, with $L$-list error probability
\begin{equation*}
P^L_e(\cC)\triangleq\PR{V\notin\cD(Y)},\qquad V\sim P_V,\ X=x_V,\ Y\sim\W(\cdot\mid X).
\end{equation*}

The PEP is defined exactly as in Section~\ref{sec:pep}, but over the enlarged candidate space, with the competing pair drawn from a \emph{uniform} auxiliary prior on the source.

\begin{definition}[PEP with source symbols]\label{def:jscc-pep}
Let $(\bar V,\bar X)\sim Q_V\cdot Q_{X|V}$ with $Q_V=\Unif{\cV}$. For a transmitted pair $(v,x)$, output $y$, and independent dither $U\sim\uU$,
\begin{align*}
\PEP{v,x}{y}&\triangleq\PR{m(\bar V,\bar X,y)>m(v,x,y)}\\
&\quad+U\cdot\PR{m(\bar V,\bar X,y)=m(v,x,y)}.
\end{align*}
\end{definition}

The uniform $Q_V$ symmetrizes comparisons across competing symbols; the true (possibly non-uniform) source law $P_V$ enters only through the operational error probability and through the metric, not through the auxiliary draw. The competing pair ranges over all of $\cV\times\cX$ including $\bar V=v$, and the dithered tie term absorbs the self-comparison.

\subsection{Achievability}

Two new complications relative to Section~\ref{sec:channel} appear. First, because the metric depends on $v$, the per-competitor PEPs are heterogeneous, so no single-PEP symmetry is available; uniform averaging over $\bar V$ handles this. Second, an $L$-list error requires at least $L$ of the heterogeneous Bernoulli comparison events to succeed, which is no longer a plain binomial. A heterogeneous-Bernoulli tail bound supplies the factor $\binom{|\cV|}{L}$ and the $L$-th power of the average PEP.

\begin{theorem}[JSCC achievability]\label{thm:jscc-rc}
For any prior $Q_{X|V}$,
\begin{align}
\bar P_e^L(Q_{X|V})
&\le\Es{P_V\cdot Q_{X|V}\cdot\W}{\min\BRAs{1,\;\tbinom{|\cV|}{L}\PEP{V,X}{Y}^L}}\label{eq:jscc-rc-bin}\\
&\le\Es{P_V\cdot Q_{X|V}\cdot\W}{\min\BRAs{1,\;\BRA{\tfrac{e|\cV|}{L}\PEP{V,X}{Y}}^{\!L}}},\label{eq:jscc-rc-ub}
\end{align}
where $\bar P_e^L(Q_{X|V})=\E{P^L_e(\cC)}$ averages over the random codebook.
\end{theorem}

\begin{IEEEproof}[Proof sketch]
An $L$-list error occurs iff at least $L$ of the competing symbols beat the transmitted one. Enlarging the competitor index set from $\cV\setminus\{v\}$ to $\cV$ keeps the uniform-average PEP equal to $\PEP{v,x}{y}$; the heterogeneous-Bernoulli tail bound at level $L$ then gives the $\binom{|\cV|}{L}\PEP^L$ factor pointwise, expectations factoring because, given $(v,x,y)$, the competing pairs are i.i.d.\ $Q_V\times Q_{X|V}$. Clipping at $1$ and using $\binom{n}{k}\le(en/k)^k$ yields \eqref{eq:jscc-rc-ub}; see Appendix~\ref{app:jscc} for the full argument.
\end{IEEEproof}

At $L=1$ the binomial bound \eqref{eq:jscc-rc-bin} collapses to the refined union bound $\RCUp$ of Section~\ref{sec:channel}. Unlike the channel-coding case, no exact random-coding identity is known here: heterogeneity of the per-competitor PEPs precludes the product form, and the binomial relaxation \eqref{eq:jscc-rc-ub} costs a $\log e$ slack.

\subsection{JSCC error spectrum}

The governing quantity is again the distribution of the randomized PEP.

\begin{definition}[JSCC error spectrum]\label{def:jscc-spectrum}
For a prior $Q_{X|V}$,
\begin{equation*}
\Fspecjscc{z}\triangleq\PRs{P_V\cdot Q_{X|V}\cdot\W}{-\log\PEP{V,X}{Y}\le z}.
\end{equation*}
\end{definition}

As in Section~\ref{sec:pep}, $\Fspecjscc{z}$ is continuous (\cref{lem:atomless}). The relaxed bound \eqref{eq:jscc-rc-ub} admits the spectral integral form
\begin{equation*}
\tilde P^L(R;Q_{X|V})\triangleq L\,e^{LR}\!\int_R^\infty\!\Fspecjscc{z}\,e^{-Lz}\,dz,
\end{equation*}
with the effective rate identified by $e^R=e|\cV|/L$: the source alphabet plays the codebook's role and $\log|\cV|$ the rate's, and since the source alphabet is fixed, $R$ is a bookkeeping parameter that matches the channel-coding formulas. The slope identity \eqref{eq:recap-slope} carries over as $\Fspecjscc{R}/\tilde P^L(R)=1+\dot E(R)/L$ with $-L<\dot E(R)\le0$, reducing to the channel-coding relation at $L=1$; this general-$L$ integral/exponent identity is proved as Lemma~\ref{app:lem-integral-L} in Appendix~\ref{app:jscc}. The spectrum $F_V$ is intrinsic to the source and channel; the list size acts purely as a horizontal $-\log L$ shift, so $\log$-list-size and rate are interchangeable.

\subsection{Converse and fixed-code exactness}

List decoding with dithered tie-breaking is exactly a threshold test on the PEP, which yields the JSCC/list generalization of the fixed-code exactness of Section~\ref{sec:channel}.

\begin{theorem}[Fixed-code exactness for list decoding]\label{thm:jscc-converse}
For any codebook $\cC=\BRAs{x_v:v\in\cV}$,
\begin{equation*}
P^L_e(\cC)=\PR{\PEP{V,X}{Y}\ge L/|\cV|},\qquad(V,X,Y)\sim P_V\cdot Q^{(\cC)}_{X|V}\cdot\W,
\end{equation*}
where $Q^{(\cC)}_{X|V}(\cdot\mid v)=\delta_{x_v}$ is the point-mass prior induced
by the codebook, entering \emph{both} the sampling law and the PEP's competitor
prior $Q_{V,X}=Q_V\cdot Q^{(\cC)}_{X|V}$. Consequently the optimal list-decoding
error probability under the dithered metric-$m$ decoder satisfies
\begin{equation*}
P^L_e\ \ge\ \inf_{Q_{X|V}}\Fspecjscc{\log|\cV|-\log L}.
\end{equation*}
\end{theorem}

\begin{IEEEproof}[Proof sketch]
Fix $y$. The per-index rank identity (Lemma~\ref{app:lem-rank}) over the candidate alphabet $\cJ=\cV$ with scores $d(v)=m(v,x_v,y)$ holds for each fixed $v$ over the dithers alone, so it may be averaged under the (non-uniform) posterior of $V$ given $Y=y$ --- only the competitor draw need be uniform. The conditional list error thus equals $\PR{\PEP{V,x_V}{y}>L/|\cV|}$; identifying $X=x_V$ and averaging over $Y$, with $>$ promoted to $\ge$ by continuity (\cref{lem:atomless}), gives the equality. The infimum bound follows since the fixed code realizes the prior $Q^{(\cC)}_{X|V}$; the effective rate is lowered by $\log L$, the extra ``room'' afforded by list decoding. See Appendix~\ref{app:jscc} for the full argument.
\end{IEEEproof}

\begin{remark}[Matched-MAP optimality]\label{rem:jscc-map}
The top-$L$ rule minimizes the list error probability under the matched MAP
metric $m(v,x,y)=P_V(v)\W(y\mid x)$, or any monotone transform of it. The
bound of \cref{thm:jscc-converse} therefore lower-bounds the \emph{optimal}
list error in the matched case; for other metrics it is a converse for the
dithered metric-$m$ decoder.
\end{remark}

The achievability rate $R_{\mathrm{ach}}=1+\log|\cV|-\log L$ and the converse rate $R_{\mathrm{conv}}=\log|\cV|-\log L$ differ by exactly one nat, the $\log e$ kernel slack of the binomial relaxation.

\subsection{Meta-converse and relation to known bounds}

For the matched MAP likelihood-ratio metric $m(v,x,y)=|\cV|\,W(y\mid x)P_V(v)/Q_Y(y)$, the spectrum is a binary-hypothesis-testing functional, the JSCC form of the meta-converse identity of Section~\ref{sec:pep}.

\begin{theorem}[JSCC meta-converse]\label{thm:jscc-meta}
With $Q_{V,X}\triangleq Q_V\cdot Q_{X|V}$, $Q_V=\Unif{\cV}$,
\begin{equation*}
\Fspecjscc{z}=\sup_{Q_Y}\BETA{1-e^{-z}}{Q_{V,X}\times Q_Y}{P_V\cdot Q_{X|V}\cdot\W}.
\end{equation*}
\end{theorem}

\begin{IEEEproof}
For the matched metric, $(Q_V\times Q_{X|V}\times Q_Y)\cdot m=P_V\cdot Q_{X|V}\cdot\W$ and the metric is a monotone transform of $W(y\mid x)P_V(v)$, so applying the error-tail meta-converse identity \eqref{eq:meta-converse-e} of Theorem~\ref{thm:meta-converse} over the candidate space $\cV\times\cX$ gives the $\BETA{}{}{}$ form. Hypotheses (A1)--(A2) hold once $Q_Y$ has full support: then $m<\infty$ everywhere, and direct summation gives $\Es{Q_{V,X}\times Q_Y}{m}=1$. The supremum may be so restricted, since the metric-tilted measure $P_V\cdot Q_{X|V}\cdot\W$ does not depend on $Q_Y$ and $\BETA{\alpha}{P}{Q}$ is continuous in $P$ on the finite alphabet.
\end{IEEEproof}

Combined with Theorem~\ref{thm:jscc-converse}, this shows the meta-converse is tight for lossless JSCC with list decoding, extending V\'azquez-Vilar's $L=1$ exactness \cite{vazquez2016bayesian} and paralleling the minimax channel converse \cite[Thm.~27]{polyanskiy2010channel}. The achievability bound \eqref{eq:jscc-rc-ub} recovers a Feinstein/Han information-spectrum direct theorem \cite{han2003information} and the Campo \etal\ RCU bound \cite{campo2011random}; the $\log L$ rate offset reproduces the list-size tradeoffs of Merhav \cite{merhav2014list} and of Tan--Moulin \cite{tan2014second}. The setting here is lossless JSCC, in the error-probability tradition whose exponent form goes back to Csisz\'ar~\cite{csiszar1980joint}; the finite-blocklength \emph{lossy} JSCC problem is treated by Kostina and Verd\'u~\cite{kostina2013jscc}, whose bounds the present ones complement rather than subsume.

In the matched case $\Fspecjscc{z}$ is convex in $Q_{X|V}$ (\cref{thm:joint-convex} on the enlarged candidate space), so the prior-optimized converse is again a convex --- indeed linear --- program: the reverse-channel route of~\cite{elkayampep1} applies verbatim over $\cV\times\cX$, and we do not redevelop it here.

\section{Extension: Channel Coding with an Erasure Option}\label{sec:erasure}

The PEP primitive of \cref{sec:pep,sec:channel} extends without modification to
decoders that may refuse to commit. A decoder with an \emph{erasure} option
either outputs a message or declares an erasure symbol $\bot$. It can fail in
two ways: an \emph{erasure} (no commitment) and an \emph{undetected error}
(commitment to a wrong message).

The structural fact is that the \emph{achievability} side of this regime is
governed by the \emph{same} error spectrum $\Fspec{Q_X}{\cdot}$ of
\cref{sec:pep}, read at \emph{two} thresholds: a boosted rate $R+\gamma$ for the
erasure event and the operating rate $R$ for the undetected-error event, where
$\gamma\ge 0$ is a confidence margin. The single parameter $\gamma$ trades
erasure rate against undetected-error rate, and Forney's classical
exponents~\cite{forney1968exponential} arise as the asymptotic limit;
finite-blocklength random-coding union bounds with and without erasures are
studied by Haim, Kochman, and Erez~\cite{HaimKE18}. On the converse side we
establish two facts. The same spectrum, read on its \emph{upper} tail
at level $R$, yields the high-rate correct-decoding probability together with an
exact fixed-code converse of Arimoto type (\cref{subsec:era-cd}). At a fixed
code's \emph{empirical} prior, by contrast, the spectrum above the operating
rate degenerates to an affine function of the error probability, so no erasure
converse can be read off it; we record this pitfall in \cref{subsec:era-empirical}.

Throughout we reuse the recap facts
\eqref{eq:recap-rcu}--\eqref{eq:recap-converse} and
\cref{thm:pep-uniform,lem:atomless}. Short proofs are given in full; the longer
correct-decoding proofs are collected in Appendix~\ref{app:era}.

\subsection{Setup and decision rule}

Fix a channel $\W:\cX\to\cY$, a codebook $\cC\subset\cX$ of size
$M=\lceil e^{R}\rceil$, and a metric $m:\cX\times\cY\to[0,\infty]$ as in
\cref{sec:channel}. Write $X_i$ for the codeword of message $i\in[M]$ and recall
the dithered PEP $\PEPU{x}{y}{u}$ of \cref{sec:pep}.

\begin{definition}[Two-threshold erasure rule]\label{def:era-rule}
Fix a confidence margin $\gamma\ge 0$ and a reference prior $Q_X$, a parameter
of the rule against which the PEP is computed. On observing $y$, set
$C_i\triangleq-\log\PEPU{X_i}{y}{U_i}$, the PEP taken relative to $Q_X$. The
decoder $\cD$ outputs message $i$ iff
\begin{enumerate}
\item $C_i\ge R+\gamma$ \quad(\emph{absolute confidence}), and
\item $C_i\ge C_j+\gamma$ for all $j\ne i$ \quad(\emph{margin over runner-up}).
\end{enumerate}
Otherwise it outputs the erasure symbol $\bot$.
\end{definition}

For a code with equiprobable messages and transmitted index $I\sim\Unif{[M]}$,
the two failure probabilities are
\begin{align}
P_{\mathrm{era}}(\cC)&\triangleq\PR{\cD(Y)=\bot},\label{eq:era-defera}\\
P_{\mathrm{ue}}(\cC)&\triangleq\PR{\cD(Y)\ne I,\ \cD(Y)\ne\bot}.\label{eq:era-defue}
\end{align}
The first clause of \cref{def:era-rule} demands that the accepted codeword have
PEP below $e^{-(R+\gamma)}$; the second demands that no competitor come within a
factor $e^{-\gamma}$ of it. Larger $\gamma$ is more conservative: it raises
$P_{\mathrm{era}}$ and lowers $P_{\mathrm{ue}}$. At $\gamma=0$ the margin clause
selects the maximum-metric candidate; the confidence clause can still erase
(the top PEP may exceed $e^{-R}$), but under the code's empirical prior with
distinct metric scores it is automatic, and the rule coincides with the
ordinary maximum-metric decoder of \cref{sec:channel}: $P_{\mathrm{era}}=0$
and $P_{\mathrm{ue}}$ is the error probability \eqref{eq:recap-converse}. The reference prior matters: the achievability
analysis of \cref{thm:era-ach} runs the rule with the \emph{ensemble} prior
$Q_X$ from which the codebook is drawn. Running the same rule under the
\emph{empirical} prior $Q_X^{\cC}$ of the code at hand is a different decoder
for the same codebook --- and, as \cref{subsec:era-empirical} shows, a degenerate
one.

\subsection{Achievability}

Both failure modes are read off the recap functional
\begin{equation}\label{eq:era-rcu}
\tilde P_e(R;Q_X)\triangleq\Es{X,Y,U}{\min\BRA{1,e^{R}\PEPU{X}{Y}{U}}}
=e^{R}\!\int_R^\infty e^{-z}\,\Fspec{Q_X}{z}\,dz
\end{equation}
of \eqref{eq:recap-rcu}, evaluated at the two thresholds $R+\gamma$ and $R$.

\begin{theorem}[Achievability with erasure option]\label{thm:era-ach}
For any prior $Q_X$ and margin $\gamma\ge 0$, the i.i.d.\ $Q_X$ random-coding
ensemble at rate $R$ satisfies
\begin{align}
\Es{\cC}{P_{\mathrm{era}}(\cC)}&\le\tilde P_e(R+\gamma;Q_X),\label{eq:era-bound-era}\\
\Es{\cC}{P_{\mathrm{ue}}(\cC)}&\le e^{-\gamma}\,\tilde P_e(R;Q_X).\label{eq:era-bound-ue}
\end{align}
\end{theorem}

\begin{IEEEproof}[Proof sketch]
Condition on the transmitted index; by \cref{thm:pep-uniform} the competitor
scores $C_i=-\log\PEPU{X_i}{Y}{U_i}$ are conditionally i.i.d.\ with the exact
tail $\PR{C_i\ge c}=e^{-c}$. An erasure requires the transmitted score $Z$ to
miss the confidence bar $R+\gamma$ or a competitor within the margin $\gamma$;
an undetected error requires a competitor clearing the \emph{larger} of the
confidence bar $R+\gamma$ and the margin bar $Z+\gamma$ --- so the union
bound reads the spectrum at $R+\gamma$ in the first case and at $R$,
suppressed by $e^{-\gamma}$, in the second (the split is at $Z=R$, not
$R+\gamma$). Integration by parts as in \eqref{eq:era-rcu} closes both
bounds; the full computation is in Appendix~\ref{app:era}.
\end{IEEEproof}

\Cref{thm:era-ach} bounds the two failure modes only \emph{on average} over the
ensemble; a code that is good for one average need not be good for the other. A
Markov argument on a weighted sum extracts a single code satisfying both bounds
simultaneously, at the price of a constant factor.

\begin{corollary}[A single code satisfying both bounds]\label{cor:era-single}
For any prior $Q_X$ and margin $\gamma\ge 0$, there exists a deterministic code
$\cC$ of rate $R$ such that
\begin{equation*}
P_{\mathrm{era}}(\cC)\le 2\,\tilde P_e(R+\gamma;Q_X)
\qquad\text{and}\qquad
P_{\mathrm{ue}}(\cC)\le 2\,e^{-\gamma}\,\tilde P_e(R;Q_X).
\end{equation*}
More generally, for any $a\in(0,1)$ the two factors of $2$ may be replaced by
$1/a$ and $1/(1-a)$ respectively, trading tightness in one failure mode against
the other.
\end{corollary}

\begin{IEEEproof}
Write $B_1\triangleq\tilde P_e(R+\gamma;Q_X)$ and
$B_2\triangleq e^{-\gamma}\tilde P_e(R;Q_X)$, and suppose first $B_1,B_2>0$. For
the weighted sum
$T(\cC)\triangleq a\,P_{\mathrm{era}}(\cC)/B_1+(1-a)\,P_{\mathrm{ue}}(\cC)/B_2$,
\cref{thm:era-ach} gives $\Es{\cC}{T(\cC)}\le a+(1-a)=1$, so some code in the
ensemble has $T(\cC)\le1$; both terms being non-negative, that code satisfies
$P_{\mathrm{era}}(\cC)\le B_1/a$ and $P_{\mathrm{ue}}(\cC)\le B_2/(1-a)$, and
$a=1/2$ gives the displayed form. If instead $B_1=0$ (resp.\ $B_2=0$), then
$P_{\mathrm{era}}(\cC)=0$ (resp.\ $P_{\mathrm{ue}}(\cC)=0$) for almost every code
in the ensemble; restricting to that almost-sure event, some code has the other
failure probability at most its ensemble average, the minimum being at most the
mean.
\end{IEEEproof}

\begin{remark}[The margin as a single knob]\label{rem:era-knob}
The bounds \eqref{eq:era-bound-era}--\eqref{eq:era-bound-ue} have a transparent
reading: the undetected-error probability is the ordinary $\RCUp$ bound
suppressed by the multiplicative factor $e^{-\gamma}$, while the erasure
probability is the same bound at the boosted rate $R+\gamma$. Thus $\gamma$ is a
single tuning knob trading throughput against reliability, and both failure
modes are kernel/threshold readings of the one spectrum $\Fspec{Q_X}{\cdot}$, as
in \cref{sec:channel}. Whenever the channel is Gallager-symmetric, the
uniform-prior optimality recalled in \cref{sec:prioropt} (from
\cite{elkayampep1}) applies at \emph{both} thresholds simultaneously, so
prior selection for the erasure problem is trivial.
\end{remark}

\subsection{Fixed codes: the empirical-prior pitfall}\label{subsec:era-empirical}

The decision rule of \cref{def:era-rule} is, for each $y$, a pair of thresholds
on the rank statistic $\PEPU{\cdot}{y}{\cdot}$, so for any fixed code its
failure probabilities are exact spectrum readings, in the spirit of
\eqref{eq:recap-converse}. One might hope to turn this exactness into an
erasure converse by reading the empirical-prior spectrum at the boosted rate
$R+\gamma$. That hope is empty. Under the empirical prior every competitor's
PEP is at least $1/M$, and a short computation (Appendix~\ref{app:era})
shows that for every
$\gamma\ge0$ the boosted-rate spectrum is an affine function of the error
probability alone (converse-natural size $M=e^{R}$, \cref{rem:size-conv}),
$\Fspec{Q_X^{\cC}}{R+\gamma}=1-e^{-\gamma}\bigl(1-P_e(\cC)\bigr)$;
equivalently, the two-threshold rule commits with the data-independent
probability $e^{-\gamma}$ (up to tie-block corrections), so that
$P_{\mathrm{era}}(\cC)=1-e^{-\gamma}$ and
$P_{\mathrm{ue}}(\cC)=e^{-\gamma}P_e(\cC)$. The empirical-prior spectrum above
the operating rate is a dither artifact, which is why the achievability
analysis of \cref{thm:era-ach} runs the rule at the \emph{ensemble} prior,
where the spectrum is informative. The only fixed-code converse content
available at the empirical prior is the channel-coding identity
\eqref{eq:recap-converse} at rate $R$; a converse for the
erasure/undetected-error trade-off over all decision rules remains open to our
knowledge, even at the exponent level.

\subsection{High-rate correct decoding and a strong converse}\label{subsec:era-cd}

Error analysis reads the \emph{lower} tail of the PEP spectrum. The
complementary question---how likely is \emph{correct} decoding?---reads the
\emph{upper} tail, and is the regime where union-bound techniques are
ineffective: at high rates the error probability approaches $1$ while the
correct-decoding probability still carries exponential structure. We place the
two on the same footing.

Recall the exact random-coding identity \eqref{eq:recap-exact-rc};
working under the achievability-natural convention $M-1=e^{R}$
(\cref{rem:size-conv}), the correct-decoding probability is its complement,
\begin{equation}\label{eq:era-pc-exact}
\bar P_c(R;Q_X)\triangleq 1-\bar P_e(R;Q_X)
=\Es{X,Y,U}{\BRA{1-\PEPU{X}{Y}{U}}^{M-1}}.
\end{equation}
The governing object is now the upper tail of the PEP.

\begin{definition}[Correct-decoding spectrum]\label{def:era-cd-spectrum}
The \emph{upper-tail spectrum} is
\begin{equation}\label{eq:era-fc}
F_c(Q_X;z)\triangleq\PR{-\log\PEPU{X}{Y}{U}\ge z},
\qquad(X,Y,U)\sim Q_X\cdot\W\times\uU.
\end{equation}
\end{definition}

The upper-tail spectrum equals $1-\Fspec{Q_X}{z}$ at every continuity point of
the spectrum \eqref{eq:def-spectrum}; both tails are defined with closed
inequalities, so in general they differ by the atom mass at $z$, and under the
atomlessness of \cref{lem:atomless} the identity holds for every $z$.
Written against this spectrum, the exact identity \eqref{eq:era-pc-exact}
becomes an integral against a kernel concentrating near $R$.

\begin{theorem}[Integral form and high-rate localization]\label{thm:era-cd-ach}
With $g(z)\triangleq(1-e^{-z})^{M-1}$ and $M-1=e^{R}$,
\begin{equation}\label{eq:era-pc-integral}
\bar P_c(R;Q_X)=\int_0^\infty F_c(Q_X;z)\,g'(z)\,dz,\qquad
g'(z)=(M-1)(1-e^{-z})^{M-2}e^{-z},
\end{equation}
and for every $\gamma>0$,
\begin{equation}\label{eq:era-pc-localized}
(1-e^{-\gamma})\,F_c(Q_X;R+\gamma)\ \le\ \bar P_c(R;Q_X)\ \le\
F_c(Q_X;R-\gamma)+e^{-e^{\gamma}}.
\end{equation}
In particular, taking $\gamma=\log R$ for $R>1$,
\begin{equation}\label{eq:era-pc-tuned}
\BRA{1-\tfrac1R}F_c(Q_X;R+\log R)\ \le\ \bar P_c(R;Q_X)\ \le\
F_c(Q_X;R-\log R)+e^{-R}.
\end{equation}
\end{theorem}

\begin{IEEEproof}[Proof sketch]
Substituting $Z=-\log\PEPU{X}{Y}{U}$ in \eqref{eq:era-pc-exact} gives
$\bar P_c=\E{g(Z)}$; since $g(0)=0$ and $g(\infty)=1$, a layer-cake argument
transfers the differential onto the kernel, yielding
\eqref{eq:era-pc-integral}. The kernel $g'$ concentrates near
$z\approx\log(M-1)=R$, so splitting the integral at $R\pm\gamma$ and using
monotonicity of $F_c$ gives \eqref{eq:era-pc-localized}: the multiplicative
shrinkage $(1-e^{-\gamma})$ on the lower side is Bernoulli's inequality, and the
doubly-exponential tail $e^{-e^{\gamma}}$ on the upper side comes from
$(1-u)^{M-1}\le e^{-(M-1)u}$ applied on the small-$z$ piece with
$u=e^{-(R-\gamma)}$, so that $(M-1)u=e^{\gamma}$. Then $1-e^{-\log R}=1-1/R$
and $e^{-e^{\log R}}=e^{-R}$ give \eqref{eq:era-pc-tuned}. Full proofs in
Appendix~\ref{app:era}.
\end{IEEEproof}

The additive $e^{-R}$ is negligible whenever $F_c(Q_X;R+\log R)$ is not
exponentially small, so $\bar P_c$ is pinned to the upper-tail spectrum within a
factor $1-1/R\approx 1$. Unlike error analysis, the upper tail is not subject to
the straight-line phenomenon~\cite{shannon1967lower}.

The converse for correct decoding is again \emph{exact}, mirroring the
fixed-code exactness of \eqref{eq:recap-converse}.

\begin{theorem}[Fixed-code correct-decoding converse]\label{thm:era-cd-converse}
For any codebook $\cC$ of size $M$ with empirical prior $Q_X^{\cC}$ and
$R=\log M$,
\begin{equation}\label{eq:era-pc-converse}
P_c(\cC)=F_c(Q_X^{\cC};R),\qquad\text{hence}\qquad
P_c(R)\le\sup_{Q_X}F_c(Q_X;R),
\end{equation}
where $P_c(R)\triangleq\sup_{\cC:\,|\cC|=M}P_c(\cC)$ is the best
correct-decoding probability among size-$M$ codes.
\end{theorem}

\begin{IEEEproof}
By the rank interpretation (\cref{thm:pep-uniform}), the maximum-metric decoder
is correct iff the transmitted codeword has rank $1$, i.e.\
$\PEPU{X}{Y}{U}\le 1/M=e^{-R}$. Taking complements,
$P_c(\cC)=\PR{-\log\PEPU{X}{Y}{U}\ge R}=F_c(Q_X^{\cC};R)$. The supremum over the
induced prior gives the code-independent bound.
\end{IEEEproof}

Combining \eqref{eq:era-pc-tuned} with \eqref{eq:era-pc-converse} sandwiches the
optimal correct-decoding probability,
\begin{equation}\label{eq:era-pc-sandwich}
\BRA{1-\tfrac1R}\sup_{Q_X}F_c(Q_X;R+\log R)\ \le\ P_c(R)\ \le\
\sup_{Q_X}F_c(Q_X;R),
\end{equation}
so the prior-optimized upper-tail spectrum near $R$ fully determines high-rate
performance. As in \cref{sec:channel}, this spectrum admits a reverse-channel
characterization in the matched case.

\begin{theorem}[Variational form, matched case]\label{thm:era-cd-meta}
Let $P_Y(y)=\sum_x Q_X(x)\W(y\mid x)$. For the matched metric (rank-equivalent
to $\log\W(y\mid x)$) and every $z>0$,
\begin{equation}\label{eq:era-cd-meta}
F_c(Q_X;z)=e^{-z}\sup_{D_\infty(W^*\|Q_X\mid P_Y)\le z}
\sum_{x\in\cX}\ \sum_{y\in\mathrm{supp}(P_Y)}W^*(x\mid y)\,\W(y\mid x),
\end{equation}
the supremum over reverse channels $W^*:\cY\to\cX$ obeying the divergence cap.
\end{theorem}

\begin{IEEEproof}[Proof sketch]
This is the upper-tail instance of the reverse-channel identity
\eqref{eq:rc-sup} at the matched metric $m(x,y)=\W(y\mid x)/P_Y(y)$: the $P_Y$
factor cancels against the reverse-channel measure $P_Y\cdot W^*$, leaving the
stated correlation form. The full verification is in Appendix~\ref{app:era}.
\end{IEEEproof}

For $z\le0$ the identity is trivial: $F_c(Q_X;z)=1$, since the PEP lies in
$[0,1]$.

\begin{remark}[Arimoto's strong-converse exponent]\label{rem:era-arimoto}
Specializing \eqref{eq:era-cd-meta} to the matched-ML metric and applying the
data-processing inequality for R\'enyi divergence recovers Arimoto's
strong-converse bound~\cite{arimoto1973converse}; the non-asymptotic form of
Polyanskiy and Verd\'u~\cite{polyanskiy2010arimoto} states that any
$(M,\varepsilon)$ code satisfies
$d_\lambda\!\BRA{1-\varepsilon\,\|\,1/M}\le K_\lambda(X;Y)$ for $\lambda>0$,
$\lambda\ne1$, with $K_\lambda$ Sibson's information radius. Taking
$\lambda=1/(1+\rho)$ with $\rho\in(-1,0)$ yields
$P_c(\cC)\le\exp\{\rho R-E_0(\rho,Q_X)\}$ for a size-$e^{R}$ code with empirical
prior $Q_X$ --- decreasing in $R$, as a strong converse must be. The upper-tail
variational form \eqref{eq:era-cd-meta} reaches the same conclusion through the
framework's reverse-channel identity.
\end{remark}

\section{Extension: The Two-User Multiple-Access Channel}\label{sec:network}

The previous extensions enlarged the candidate space but kept a single
transmitter. The two-user multiple-access channel (MAC) is the boundary case of
the framework: one of the two structural themes of~\cite{elkayampep1} survives intact, the other
breaks. \emph{Comparability} survives --- both achievability and converse are
read off PEP spectra, a sum-spectrum for achievability and a max-spectrum for the
exact fixed-code converse, exactly as in \cref{sec:channel}. \emph{Joint
convexity} (\cref{thm:joint-convex}) does not: the two encoders choose their
priors independently, and the joint error event couples them \emph{bilinearly},
so the single linear program of \cref{sec:prioropt} does not extend. We make both
statements precise and explain the obstruction.

As in \cref{sec:jscc} the
construction reuses the uniformity of \cref{thm:pep-uniform} and the fixed-code
identity behind \cref{eq:recap-converse} verbatim, instantiated now on three
candidate spaces at once.

\subsection{Setup}

Fix correlated finite sources $(V_1,V_2)\sim P_{V_1,V_2}$ on $\cV_1\times\cV_2$, a
MAC $W_{Y\mid X_1,X_2}$ with inputs $\cX_1,\cX_2$ and output $\cY$, and a decoding
metric $m:\cV_1\times\cX_1\times\cV_2\times\cX_2\times\cY\to[0,\infty)$. The two
encoders draw codewords independently: $X_1(v_1)\sim Q_{X_1\mid V_1}(\cdot\mid
v_1)$ and $X_2(v_2)\sim Q_{X_2\mid V_2}(\cdot\mid v_2)$. On source pair
$(V_1,V_2)$, terminal $i$ transmits $X_i(V_i)$ and the receiver observes
$Y\sim W(\cdot\mid X_1,X_2)$. The joint decoder attaches independent dithers
$U_{v_1,v_2}\sim\uU$, one per candidate pair, and outputs the pair maximal in the
weak lexicographic order on $\BRA{m(v_1,X_1(v_1),v_2,X_2(v_2),y),\,U_{v_1,v_2}}$,
as in \cref{sec:pep}. The unique-decoding error probability of a (possibly random)
code $\cC$ is $P_e(\cC)=\PR{(\hat V_1,\hat V_2)\neq(V_1,V_2)}$, and
$\bar P_e=\Es{\cC}{P_e(\cC)}$ for the random-coding ensemble.

Specializing recovers the classical settings: the identity channel
$Y=(X_1,X_2)$ gives distributed (Slepian--Wolf) source
coding~\cite{slepian1973noiseless,chen2020lossless}; independent uniform sources
on $[M_1]\times[M_2]$ give MAC channel coding at rates $(\log M_1,\log
M_2)$~\cite{ahlswede1971multi,liao1972multiple}.

\subsection{Three pairwise error classes}

A competing pair $(v_1',v_2')$ that differs from the transmitted $(v_1,v_2)$ falls
in exactly one of three classes --- first coordinate only, second coordinate
only, or both --- and each calls for its own PEP, defined against the uniform
auxiliary priors $Q_{V_i}=\Unif{\cV_i}$ so that \cref{thm:pep-uniform} applies on
each space.

\begin{definition}[Row, column, and joint PEPs]\label{def:mac-pep}
Let $(\bar V_1,\bar X_1)\sim Q_{V_1}\cdot Q_{X_1\mid V_1}$ and
$(\bar V_2,\bar X_2)\sim Q_{V_2}\cdot Q_{X_2\mid V_2}$ be drawn independently. For
a transmitted tuple $(v_1,x_1,v_2,x_2,y)$ and dither $U\sim\uU$, the
\emph{message-1 (row)}, \emph{message-2 (column)}, and \emph{joint} PEPs are the
dithered comparison probabilities
\begin{align*}
\PEP{v_1,x_1}{v_2,x_2,y}&\triangleq\PRs{\bar V_1,\bar X_1}{m(\bar V_1,\bar X_1,v_2,x_2,y)\,{>}\,m}
 +U\,\PRs{\bar V_1,\bar X_1}{m(\bar V_1,\bar X_1,v_2,x_2,y)\,{=}\,m},\\
\PEP{v_2,x_2}{v_1,x_1,y}&\triangleq\PRs{\bar V_2,\bar X_2}{m(v_1,x_1,\bar V_2,\bar X_2,y)\,{>}\,m}
 +U\,\PRs{\bar V_2,\bar X_2}{m(v_1,x_1,\bar V_2,\bar X_2,y)\,{=}\,m},\\
\PEP{v_1,x_1,v_2,x_2}{y}&\triangleq\PRs{\bar V_1,\bar X_1,\bar V_2,\bar X_2}{m(\bar V_1,\bar X_1,\bar V_2,\bar X_2,y)\,{>}\,m}\\
 &\qquad +U\,\PRs{\bar V_1,\bar X_1,\bar V_2,\bar X_2}{m(\bar V_1,\bar X_1,\bar V_2,\bar X_2,y)\,{=}\,m},
\end{align*}
where $m=m(v_1,x_1,v_2,x_2,y)$ is the transmitted score.
\end{definition}

By \cref{thm:pep-uniform} applied to each candidate space, each randomized PEP is
$\uU$-distributed under its own uniform prior. Writing
$(V_1,V_2)\sim P_{V_1,V_2}$, $X_i\mid V_i\sim Q_{X_i\mid V_i}$, and
$Y\sim W(\cdot\mid X_1,X_2)$, define the three scaled error components
\begin{equation}\label{eq:mac-components}
\begin{gathered}
E_1=|\cV_1|\,\PEP{V_1,X_1}{V_2,X_2,Y},\\
E_2=|\cV_2|\,\PEP{V_2,X_2}{V_1,X_1,Y},\\
E_{1,2}=|\cV_1||\cV_2|\,\PEP{V_1,X_1,V_2,X_2}{Y}.
\end{gathered}
\end{equation}

\subsection{Achievability via the sum-spectrum}

Union-bounding over the three competitor classes, each averaged against its own
uniform reference, gives the direct bound.

\begin{theorem}[MAC random-coding bound]\label{thm:mac-rc}
For any priors $Q_{X_1\mid V_1},Q_{X_2\mid V_2}$ and any metric $m$ with dithered
tie-breaking,
\begin{equation}\label{eq:mac-rc}
\bar P_e\le\E{\min\BRA{1,\;E_1+E_2+E_{1,2}}},
\end{equation}
the expectation over $(V_1,V_2,X_1,X_2,Y)$.
\end{theorem}

\begin{IEEEproof}[Proof sketch]
Condition on the transmitted tuple and dither. An error needs a competitor
beating the truth; union-bounding over the three classes gives three sums. For
the row class, a four-step chain bounds the row sum by
$|\cV_1|\,\PEP{v_1,x_1}{v_2,x_2,y}$, whence $\E{A}\le E_1$: extend the index set
from $\cV_1\setminus\{v_1\}$ to $\cV_1$ (only enlarging the bound); apply the
lexicographic-order equivalence (\cref{thm:pep-uniform}(i)); average the
competitor under the uniform $Q_{V_1}$ (yielding the factor $|\cV_1|$); and
apply uniformity (\cref{thm:pep-uniform}(iii)). The column class is
symmetric, $\E{B}\le E_2$. The joint class ranges over
$(|\cV_1|-1)(|\cV_2|-1)$ pairs with independent per-pair dithers; the
\emph{same} four-step chain, applied on the product candidate space
$\cV_1\times\cV_2$ under the uniform product reference $Q_{V_1}\times Q_{V_2}$,
gives $\E{C}\le E_{1,2}$. The individual beat probabilities are \emph{not}
equal across joint-class competitors --- the metric and conditional priors depend
on the pair --- but the chain never needs them to be. Summing and clipping at $1$
yields \eqref{eq:mac-rc}; the full argument is given in
Appendix~\ref{app:mac}.
\end{IEEEproof}

The governing object is the distribution of the sum.

\begin{definition}[Sum-spectrum]\label{def:mac-sumspec}
$\displaystyle F_s(z)\triangleq\PR{E_1+E_2+E_{1,2}\ge e^{-z}}$, $z\in\mR$.
\end{definition}

\begin{corollary}[Integral form]\label{cor:mac-integral}
With $F_s$ as above,
\begin{equation}\label{eq:mac-sumspec}
\bar P_e\le\int_0^\infty F_s(z)\,e^{-z}\,dz.
\end{equation}
\end{corollary}

\begin{IEEEproof}
With $S=E_1+E_2+E_{1,2}\ge0$, the layer-cake identity gives
$\E{\min(1,S)}=\int_0^1\PR{S>w}\,dw$; the substitution $w=e^{-z}$ on $(0,1]$
turns the integrand into $\PR{S>e^{-z}}\,e^{-z}$, and $\PR{S>e^{-z}}=F_s(z)$ for
a.e.\ $z$, so no atomlessness of the sum $S$ is required: the two differ only by
the jump of the monotone map $z\mapsto F_s(z)$, nonzero on an at most countable
(hence Lebesgue-null) set.
\end{IEEEproof}

The kernel $e^{-z}$ and the integral are exactly those of the point-to-point
$\RCUp$ bound \eqref{eq:recap-rcu}; only the spectrum changed, from a single PEP
to the sum of three.

\subsection{Converse via the max-spectrum}

For a deterministic code $\cC$ the induced priors are degenerate,
$Q^{(\cC)}_{X_i\mid V_i=v_i}=\delta_{x_i(v_i)}$, and the three components
\eqref{eq:mac-components} become functions of $(V_1,V_2,Y)$ alone.

\begin{theorem}[Exact fixed-code converse]\label{thm:mac-converse}
For any fixed code $\cC$,
\begin{equation}\label{eq:mac-conv}
P_e(\cC)=\PR{E_{1,2}\ge1}
\ \ge\ \max\BRAs{\PR{E_1\ge1},\,\PR{E_2\ge1}}.
\end{equation}
\end{theorem}

\begin{IEEEproof}[Proof sketch]
The joint candidate space $\cV_1\times\cV_2$ with channel-input space
$\cX_1\times\cX_2$ is a single-transmitter (JSCC) problem, so the fixed-code
exactness of \cref{thm:jscc-converse} gives
$P_e(\cC)=\PR{\PEP{V_1,X_1,V_2,X_2}{Y}\ge1/|\cV_1||\cV_2|}=\PR{E_{1,2}\ge1}$. For
the marginal bound, compare against a \emph{row-restricted} benchmark: a genie
reveals $(V_2,X_2)$ to a decoder that runs the \emph{same} metric with the
\emph{same} dithers, restricted to the row $\{(v_1',V_2):v_1'\in\cV_1\}$. The
comparison is pathwise --- for every realization, being lexicographically
maximal over the whole grid implies being maximal over one's own row --- so
$P_e(\cC)\ge P_e^{\mathrm{row}}(\cC)$. The row-restricted problem is a
point-to-point JSCC instance for terminal~1 with conditional channel
$W(\cdot\mid\cdot,x_2)$, whose fixed-code identity gives
$P_e^{\mathrm{row}}(\cC)=\PR{E_1\ge1}$. Symmetry yields the bound with $E_2$.
The full argument is given in Appendix~\ref{app:mac}.
\end{IEEEproof}

The converse-side spectrum keeps the three components separate and takes the
largest reading.

\begin{definition}[Max-spectrum]\label{def:mac-maxspec}
\begin{equation}\label{eq:mac-maxspec}
F_m(z)\triangleq\max\BRAs{\PR{E_{1,2}\ge e^{-z}},\,\PR{E_1\ge e^{-z}},\,\PR{E_2\ge e^{-z}}},
\qquad z\in\mR.
\end{equation}
\end{definition}

By \cref{thm:mac-converse}, the fixed-code error is the max-spectrum read at the
origin, $P_e(\cC)=F_m^{(\cC)}(0)$ --- the network analogue of the channel-coding
identity $P_e(\cC)=\Fspec{Q_X^\cC}{R}$ of \eqref{eq:recap-converse}.

\subsection{Relation between achievability and converse}

Both bounds are spectral readings, so they are directly comparable. Since
$E_1,E_2,E_{1,2}\ge0$, the event $\{E_k\ge t\}$ is contained in
$\{E_1+E_2+E_{1,2}\ge t\}$ for each $k$, whence
\begin{equation}\label{eq:mac-fsfm}
F_m(z)\le F_s(z)\qquad\text{for all }z.
\end{equation}
The achievability functional therefore dominates the converse functional
pointwise: the union bound aggregates all three error events, while the converse
keeps only the dominant one. This is a genuine gap --- $\int F_s\,e^{-z}\,dz$
cannot be replaced by $\int F_m\,e^{-z}\,dz$, which would be smaller and not a
valid upper bound on $\bar P_e$.

\begin{corollary}[One-shot sandwich]\label{cor:mac-sandwich}
The optimal unique-decoding error $P_e^\star=\inf_\cC P_e(\cC)$ obeys
\begin{equation}\label{eq:mac-sandwich}
\inf_{Q_{X_1\mid V_1},\,Q_{X_2\mid V_2}}F_m(0)
\ \le\ P_e^\star\ \le\
\inf_{Q_{X_1\mid V_1},\,Q_{X_2\mid V_2}}\int_0^\infty F_s(z)\,e^{-z}\,dz.
\end{equation}
\end{corollary}

\begin{IEEEproof}
The lower bound is \cref{thm:mac-converse} optimized over priors; the upper bound
is \cref{cor:mac-integral} optimized over priors, a code attaining $\bar P_e$
existing by averaging.
\end{IEEEproof}

\subsection{The bilinear-coupling obstruction}

In the matched MAP case the components of \eqref{eq:mac-maxspec} each admit a
reverse-channel / $\beta$-functional form, exactly as the single-user spectrum
does in \cref{thm:meta-converse,thm:reverse-channel}: the joint component is the
meta-converse on the candidate space $\cV_1\times\cV_2$, and the row/column
components are the conditional meta-converses obtained by freezing the other
terminal. Each component is, for fixed priors, an affine functional of a reverse
channel subject to a $D_\infty$ cap, so the optimized success spectrum
$\sup\min_k\PR{E_k\le e^{-z}}$ (the supremum over the reverse channels in the
three variational forms) is a linear program in those channels --- the network
image of \cref{thm:joint-convex}.

The single LP of \cref{sec:prioropt} nevertheless does \emph{not} extend, and the
reason is structural rather than technical.

\begin{remark}[Why joint convexity fails]\label{rem:mac-bilinear}
The joint component's $D_\infty$ cap is taken against the reference
$Q_{V_1}\!\cdot Q_{X_1\mid V_1}\times Q_{V_2}\!\cdot Q_{X_2\mid V_2}$, which
carries the \emph{product} $Q_{X_1\mid V_1}(x_1\mid v_1)\,Q_{X_2\mid V_2}(x_2\mid
v_2)$ of the two encoder priors. As a function of the pair
$(Q_{X_1\mid V_1},Q_{X_2\mid V_2})$ this product is bilinear, not jointly convex;
the row and column components are each affine in the two priors separately and add
no coupling, but the joint component alone destroys joint convexity. The combined
problem is \emph{biconvex} --- a linear program in each encoder's prior with the
other fixed --- so \cref{thm:joint-convex} holds one terminal at a time but fails
jointly, and the polynomial-size prior-optimization LP of \cref{sec:prioropt} has
no single-program analogue here.

What remains available is a \emph{convex outer bound}. One direction of error is
safe by the structure of the problem: over-estimating the prior-optimized
success probability can only \emph{weaken} the converse, never falsely tighten
it, so any convex relaxation of the bilinear product yields a valid
computational converse and the only question is how much is lost. The
construction of such relaxations is standard and we do not develop it.
Lifting the bilinear products to a joint variable constrained by McCormick
envelopes~\cite{mccormick1976computability} yields a linear program; the
Sherali--Adams (RLT) hierarchy~\cite{sherali1990hierarchy} refines it at
increasing cost. For finite-blocklength converses in the two-terminal setting,
Jose and Kulkarni~\cite{jose2018improved} carry out exactly this
lift-and-project program for Slepian--Wolf coding --- the identity-channel
specialization of this section --- and their relaxations apply here without
change. Being
biconvex, the problem also admits alternating maximization over the two
encoders, monotone but only locally optimal. We therefore claim for the MAC the
comparable achievability/converse pair \eqref{eq:mac-fsfm}, but \emph{not} the
exact prior-optimization LP of the point-to-point case; quantifying the
relaxation slack is left open.
\end{remark}

This is the honest boundary of the framework. The PEP primitive reaches the
network setting --- three pairwise events, a sum-spectrum direct bound, an exact
max-spectrum converse, all comparable --- but the convex prior optimization that
made the single-user converse computable rests on a single prior, and two
coupled priors break it. The one-shot achievability literature for
networks~\cite{li2021unified,yassaee2013technique,verdu2012non,chen2020lossless}
is centered on direct bounds, and finite-blocklength converses and second-order
(dispersion) analyses of the MAC have been obtained by other
means~\cite{tan2014dispersions,molavianjazi2015second,jose2018improved}; what
appears to be new here is an exact fixed-code converse in the same spectral
language as the direct bound, with the convexity obstruction stated explicitly.

\section{Conclusion}\label{sec:conc}

The four developments of this paper make one point: the pairwise-error-probability
primitive of~\cite{elkayampep1} is not tied to channel coding. The same
rank-transform primitive, evaluated on an enlarged candidate space, gave a
random-coding achievability bound and an \emph{exact} fixed-code converse for
lossy source coding, joint source--channel coding with list decoding, and the
two-user multiple-access channel; for erasure/undetected-error decoding it gave
the achievability pair together with an exact correct-decoding converse. The
primitive is read on \emph{events} for channel, erasure, and multiuser coding, and
on \emph{values} for distortion. In each single-terminal setting the matched
case retains the convexity that makes prior optimization a linear program, and
we indicated per setting how the prior-optimization program extends without
redeveloping it --- the computational story remains that of~\cite{elkayampep1}.

The multiuser case is the one genuine departure. The achievability/converse pair
survives through three pairwise error events, but the bilinear coupling of the
two encoder priors removes joint convexity, so the single-LP prior optimization
of the point-to-point case does not extend; the precise obstruction, and any
structural restriction on the coupling that would restore convexity, are left
open. Together with~\cite{elkayampep1}, these results extend the error-spectrum
treatment across source, channel, joint, erasure, and multiuser settings.

\appendices

\section{Framework Lemmas: Recalled and Extended}\label{app:rank-recall}
This appendix proves the framework facts this paper needs that are not proved
in~\cite{elkayampep1}: the inversion clause of the randomized
probability-integral transform (RPIT) and the tail--quantile inverse map of
\cref{thm:pep-uniform}(iv); the finite rank-equivalence identity in its
per-index form (\cite{elkayampep1} proves the uniform-index average); and the
complementary-tail meta-converse identity \eqref{eq:meta-converse-c} (stated
in~\cite{elkayampep1} but neither proved nor used there).

\begin{lemma}[RPIT, extended]\label{app:lem-rpit}
Let $T$ be a real random variable with CDF $F_T$ and atom function
$P_T(t)\triangleq F_T(t)-F_T(t^-)$, and let $U\sim\uU$ be independent of $T$.
Then $\Phi\triangleq F_T(T^-)+U\,P_T(T)\sim\uU$. Moreover, the generalized
inverse $F_T^{\langle-1\rangle}(\phi)\triangleq\inf\{t:F_T(t)\ge\phi\}$
recovers $T$ from $\Phi$: $T=F_T^{\langle-1\rangle}(\Phi)$ a.s.
\end{lemma}

\begin{IEEEproof}
Uniformity of $\Phi$ is proved in~\cite{elkayampep1}; only the inversion
claim is new. If $P_T(T)>0$ then, since $U\in(0,1)$ a.s.,
$F_T(T^-)<\Phi<F_T(T)$, and the smallest $t$ with $F_T(t)\ge\Phi$ is $T$
itself. If $P_T(T)=0$ then $\Phi=F_T(T)$, and
$F_T^{\langle-1\rangle}(\Phi)<T$ requires $F_T$ to be constant on an interval
$[t,T]$ with $t<T$ --- but a maximal flat piece $(a,b]$ of $F_T$ has
$\PR{T\in(a,b]}=F_T(b)-F_T(a)=0$, and there are countably many such pieces,
so this failure event is null.
\end{IEEEproof}

\begin{IEEEproof}[Proof of \cref{thm:pep-uniform}(iv)]
Fix $y$ and let $T\triangleq m(X,y)$ with $X\sim Q_X$, CDF $F_T$ and atom
function $P_T$. Since $G_{X,y}=1-F_T(T)$ and $H_{X,y}=P_T(T)$,
\[
\Phi\;\triangleq\;1-\PEPU{X}{y}{U}
\;=\;1-G_{X,y}-U\,H_{X,y}
\;=\;F_T(T^-)+(1-U)\,P_T(T),
\]
and $1-U\sim\uU$ is independent of $T$, so Lemma~\ref{app:lem-rpit} applied with
dither $1-U$ gives $m(X,y)=T=F_T^{\langle-1\rangle}(\Phi)
=F_T^{\langle-1\rangle}\bigl(1-\PEPU{X}{y}{U}\bigr)$ a.s. Define
$\tilde m(w,y)\triangleq F_T^{\langle-1\rangle}(1-w)$ (measurable jointly in
$(w,y)$, since $\{(w,y):\tilde m(w,y)>s\}=\{(w,y):F_T^{(y)}(s)<1-w\}$ and
$y\mapsto F_T^{(y)}(s)=Q_X\{m(\cdot,y)\le s\}$ is measurable by joint
measurability of $m$). Then $\tilde m(\PEP{X}{y},y)=m(X,y)$ a.s., and
monotonicity is immediate: $F_T^{\langle-1\rangle}$ is non-decreasing and
$w\mapsto1-w$ is decreasing, so $w\mapsto\tilde m(w,y)$ is non-increasing.
\end{IEEEproof}

The fixed-code converses of this paper use the finite-set rank identity
of~\cite{elkayampep1}, here strengthened to its per-index form: the JSCC and
network converses average it under a \emph{non-uniform} transmitted-index law,
which the per-index statement licenses.

\begin{lemma}[Finite rank equivalence, per-index form]\label{app:lem-rank}
Let $j\in[M]$ carry scores $d_j$ and i.i.d.\ dithers $U_j\sim\uU$. For each
$i\in[M]$ let $S(i)\triangleq\#\{j:d_j>d_i\}$ and $T(i)\triangleq\#\{j:d_j=d_i\}
\ge1$ (the tie block contains $i$), and set
$\PEPs{i}\triangleq\bigl(S(i)+U_i\,T(i)\bigr)/M$, the dithered PEP of index $i$
under the uniform prior on $[M]$. Let $\mathrm{rank}(i)$ be the position of
$(d_i,U_i)$ among $\{(d_j,U_j)\}_{j\in[M]}$ under $\succ$ (rank $1$ = top). Then
for every \emph{fixed} $j\in[M]$ and every $L\in\{1,\dots,M\}$,
\begin{equation}\label{eq:rank-perindex}
\PR{\PEPs{j}\le L/M}=\PR{\mathrm{rank}(j)\le L}
=\min\BRA{\frac{(L-S(j))_+}{T(j)},\,1},
\end{equation}
the probabilities over the dithers alone. Consequently, for a random index $J$
drawn from \emph{any} law on $[M]$ independent of the dithers,
$\PR{\mathrm{rank}(J)\le L}=\PR{\PEPs{J}\le L/M}$.
\end{lemma}

\begin{IEEEproof}
Condition on the scores and fix $j$; put $a\triangleq S(j)$ and
$k\triangleq T(j)\ge1$. Then
$\mathrm{rank}(j)=1+a+\#\{i\neq j:d_i=d_j,\,U_i<U_j\}$. Among the $k$ tied
dithers the rank of $U_j$ is uniform on $\{0,\dots,k-1\}$ (the dither law is
continuous, so within-block dither ties are null), whence
$\PR{\mathrm{rank}(j)\le L}=\bigl(((L-a)_+)\wedge k\bigr)/k$. On the other side
$M\,\PEPs{j}=a+U_j k$, so
$\PR{\PEPs{j}\le L/M}=\PR{U_j\le(L-a)/k}=\bigl(((L-a)_+)\wedge k\bigr)/k$, the
same value, which is the display \eqref{eq:rank-perindex}. The identity for a
random $J$ follows by averaging \eqref{eq:rank-perindex} over the law of $J$ ---
any law, since \eqref{eq:rank-perindex} holds for each fixed $j$ and the dithers
are independent of $J$; only the competitor reference inside $\PEPs{\cdot}$ need
be uniform.
\end{IEEEproof}

\begin{IEEEproof}[Proof of the complementary-tail identity \eqref{eq:meta-converse-c}]
The argument is that of the error-tail identity \eqref{eq:meta-converse-e},
proved in~\cite{elkayampep1}, with the level flipped from $1-e^{-R}$ to
$e^{-R}$; we give it in full. Abbreviate
$Q^\bullet\triangleq(Q_X\times P_Y)\cdot m$.

\emph{($\le$).} The test $T_c(x,y)\triangleq\Ind{\PEP{x}{y}\ge1-e^{-R}}$ (the
event $\{Z_c\ge R\}$) has, by uniformity (\cref{thm:pep-uniform}(iii)),
$Q_X$-mass exactly $e^{-R}$ at every fixed $y$, hence
$(Q_X\times Q_Y)(T_c=1)=e^{-R}$ for every $Q_Y$; $T_c$ is therefore feasible at
level $\alpha=e^{-R}$, and
$\BETA{e^{-R}}{Q_X\times Q_Y}{Q^\bullet}\le\Es{Q^\bullet}{T_c}
=\Es{Q_X\times P_Y}{m\,\Ind{Z_c\ge R}}$. Maximizing over $Q_Y$ gives ``$\le$''.

\emph{($\ge$).} For each $y$ pick $(\tau'_y,\theta'_y)$ solving the randomized
quantile equation
$Q_X\bigl(m(\cdot,y)<\tau'_y\bigr)+\theta'_y\,Q_X\bigl(m(\cdot,y)=\tau'_y\bigr)
=e^{-R}$. From the quantile equation, $Q_X\bigl(m(\cdot,y)\ge\tau'_y\bigr)\ge
1-e^{-R}$, so Markov's inequality gives
$\tau'_y\le\Es{Q_X}{m(X,y)}/(1-e^{-R})$, finite by (A2); hence
$\lambda^\ast\triangleq\sum_y P_Y(y)\tau'_y<\infty$ and
$Q_Y^\star(y)\triangleq P_Y(y)\tau'_y/\lambda^\ast$ is a probability measure
(if $\lambda^\ast=0$ --- which requires $Q_X\{m(\cdot,y)=0\}\ge e^{-R}$ for
every $y\in\mathrm{supp}(P_Y)$ --- both sides of the identity vanish and
there is nothing to prove). The Neyman--Pearson likelihood
ratio under $Q_Y^\star$ is
$L(x,y)=m(x,y)P_Y(y)/Q_Y^\star(y)=\lambda^\ast m(x,y)/\tau'_y$, so the optimal
$\beta$-test at level $e^{-R}$ accepts exactly where $m(\cdot,y)<\tau'_y$, with
tie fraction $\theta'_y$ at $m=\tau'_y$ --- which, since the PEP orders
candidates by descending metric, is precisely the event $\{Z_c\ge R\}$ realized
by $T_c$. Therefore
$\BETA{e^{-R}}{Q_X\times Q_Y^\star}{Q^\bullet}
=\Es{Q_X\times P_Y}{m\,\Ind{Z_c\ge R}}$, giving ``$\ge$''.
\end{IEEEproof}

\section{Rate--Distortion Proofs}\label{app:rd}

\begin{IEEEproof}[Proof of Theorem~\ref{thm:rd-rc}]
\emph{Step 1 (reduce min-distortion to a min of uniforms).} Fix $x\in\cX$ and
set $U_i\triangleq\PEC{Y_i}{x}$ for $Y_i\sim Q_Y$ i.i.d.; by uniformity
(Lemma~\ref{lem:pec-uniform}) the $U_i$ are i.i.d.\ $\uU$, so
$W_M\triangleq\min_i U_i$ has density $f_M(w)=M(1-w)^{M-1}$. The inverse map
$\tilde d(x,\cdot)$ of Lemma~\ref{lem:pec-uniform} is non-decreasing, hence
commutes with the minimum:
\begin{align*}
  \Es{Y_1,\dots,Y_M}{\min_{i}d(x,Y_i)}
  &=\Es{}{\min_i\tilde d(x,U_i)}
  =\Es{}{\tilde d(x,W_M)}\\
  &=\int_0^1\tilde d(x,w)\,f_M(w)\,dw.
\end{align*}

\emph{Step 2 (integration by parts).} Let
$A_x(w)\triangleq\Es{Q_Y}{d(x,Y)\Ind{\PEC{Y}{x}\le w}}$. Since
$\PEC{Y}{x}\sim\uU$ and $\tilde d(x,\PEC{Y}{x})=d(x,Y)$,
$A_x(w)=\int_0^w\tilde d(x,t)\,dt$, so $A_x'(w)=\tilde d(x,w)$ a.e.\ and
$A_x(0)=0$. Introduce $G_M(w)\triangleq-(1-w)^{M-1}\bigl((M-1)w+1\bigr)$, which
satisfies $G_M'(w)=M(M-1)w(1-w)^{M-2}=-w\,f_M'(w)$. Then, using $f_M(1)=0$ and
$A_x(0)=0$,
\[
  \int_0^1\tilde d(x,w)f_M(w)\,dw
  =\BRAb{A_x(w)f_M(w)}_0^1-\int_0^1 A_x(w)f_M'(w)\,dw
  =\int_0^1 w^{-1}A_x(w)\,G_M'(w)\,dw.
\]
Averaging over $X\sim P_X$ and recognizing
$w^{-1}\Es{X}{A_X(w)}=w^{-1}\Es{X,Y}{d(X,Y)\Ind{\PEC{Y}{X}\le w}}=\Dspec{z}{Q_Y}$
at $w=e^{-z}$,
\[
  D_{\mathrm{rand}}(Q_Y,R)=\int_0^1\Es{X}{w^{-1}A_X(w)}\,G_M'(w)\,dw.
\]

\emph{Step 3 (change of variables $w=e^{-z}$).} With $dw=-e^{-z}dz$ and
$G_M'(e^{-z})e^{-z}=K_M(z)$,
\[
  D_{\mathrm{rand}}(Q_Y,R)=\int_0^\infty\Dspec{z}{Q_Y}\,K_M(z)\,dz,
\]
and $\int_0^\infty K_M(z)\,dz=\int_0^1 G_M'(w)\,dw=G_M(1)-G_M(0)=1$, so $K_M$ is
a probability density.
\end{IEEEproof}

The localization stated after Theorem~\ref{thm:rd-rc} follows from the next
lemma, applied to the primitive $F(w)\triangleq\Es{X}{A_X(w)}$ at
$w_0=e^{-(R-\gamma)}$: the primitive is convex with $F(0)=0$, since $A_x$ is an
integral of the non-decreasing $\tilde d(x,\cdot)$; the inequality
$1-w\le e^{-w}$ gives
$a=(1-w_0)^{M-1}\le e^{-e^{\gamma}(1-e^{-R})}$; and $\Dspec{0}{Q_Y}\le d_{\max}$.

\begin{lemma}[Convex integral bound]\label{app:lem-convex-int}
Let $F:[0,1]\to\mR$ be convex with $F(0)=0$, fix $M>1$, $w_0\in(0,1)$, and
$a\triangleq(1-w_0)^{M-1}$. Then
\[
  \int_0^1 F(w)\,M(M-1)(1-w)^{M-2}\,dw\le\frac{F(w_0)}{w_0}(1-a)+F(1)\,a.
\]
\end{lemma}

\begin{IEEEproof}
Split at $w_0$. On $[0,w_0]$, convexity and $F(0)=0$ give
$F(w)/w\le F(w_0)/w_0$, so the left piece is at most
$\frac{F(w_0)}{w_0}\int_0^{w_0}M(M-1)w(1-w)^{M-2}dw
=\frac{F(w_0)}{w_0}[(1-a)-a(M-1)w_0]$. On $[w_0,1]$, the chord bound
$F(w)\le\frac{1-w}{1-w_0}F(w_0)+\frac{w-w_0}{1-w_0}F(1)$ and the elementary
integrals $M(M-1)\int_{w_0}^1\frac{1-w}{1-w_0}(1-w)^{M-2}dw=(M-1)a$ and
$M(M-1)\int_{w_0}^1\frac{w-w_0}{1-w_0}(1-w)^{M-2}dw=a$ give the right piece
$\le F(w_0)(M-1)a+F(1)a$. Adding the two pieces, the $a(M-1)w_0$ terms cancel,
leaving $\frac{F(w_0)}{w_0}(1-a)+F(1)a$.
\end{IEEEproof}

\begin{IEEEproof}[Proof of Theorem~\ref{thm:rd-converse}]
Fix $x$ and let $\cC=\{y_1,\dots,y_M\}$ with $Q_Y^{\cC}=\Unif{\cC}$. Ordering
distortions ascending and applying the rank identity
(Lemma~\ref{app:lem-rank}) with $d_j=d(x,y_j)$ at $L=1$, the dither-expected
weight $\Es{U}{\Ind{\PEC{j}{x}\le1/M}}$ equals $1/T_j$ on the minimal-distortion
block ($T_j$-fold tied minima) and $0$ otherwise; these weights are
nonnegative, supported on the minimizers, and sum to one. Hence
\[
\min_i d(x,y_i)=\sum_j d(x,y_j)\,\Es{U}{\Ind{\PEC{j}{x}\le1/M}}
=M\,\Es{Q_Y^{\cC}}{d(x,Y)\Ind{\PEC{Y}{x}\le1/M}}.
\]
Averaging over $X\sim P_X$
gives $D(\cC)=\Dspec{\log M}{Q_Y^{\cC}}$. Taking $M=e^{R}$ and infimizing over
$Q_Y$ yields the converse.
\end{IEEEproof}

\begin{IEEEproof}[Proof of Theorem~\ref{thm:rd-variational}]
Both
identities are instances of the recap propositions under the role transposition
that sends the observation to the source $X\sim P_X$, the codeword to the
candidate $Y\sim Q_Y$, and the PEP to $\PEC{Y}{X}$. Identity \eqref{eq:rd-beta}
is the complementary-tail $\beta$-form \eqref{eq:meta-converse-c} (proved in
Appendix~\ref{app:rank-recall}) with the auxiliary supremum on the $\cX$-side;
identity \eqref{eq:rd-dinf} is the $\inf$ reverse-channel form \eqref{eq:rc-inf},
proved in~\cite{elkayampep1}, with the reverse channel becoming the forward
test channel $\W:\cX\to\cY$. The conditional and joint $D_\infty$ coincide
because the shared $P_X$ marginal cancels in the density ratio. Convexity of
$\Dspec{z}{Q_Y}$ in $Q_Y$ follows from the role-transposed joint convexity
(Theorem~\ref{thm:joint-convex}).
\end{IEEEproof}

\begin{IEEEproof}[Proof of Lemma~\ref{lem:excess}]
Conditioning on $X=x$ with
$q=q(x)$, the indicator distortion has $\PEC{Y}{x}=U\,q$ on the branch
$d_e=0$ and $\PEC{Y}{x}=q+U(1-q)$ on the branch $d_e=1$ (probability $1-q$);
only the latter contributes to $d_e\,\Ind{\PEC{}\le w}$, and there
\[
\PR{\PEC{Y}{x}\le w\mid X=x,\,d_e=1}\cdot(1-q)=(w-q)_+.
\]
Averaging over $X$ and
multiplying by $e^{z}$ ($w=e^{-z}$) gives the first form; $(w-q)_+=\int_0^w
\Ind{q\le t}\,dt$ gives the integral form.
\end{IEEEproof}

For the Matsuta--Uyematsu identity used in \cref{subsec:excess},
\[
D_\infty(P_X\W\|P_X\times Q_Y)
=\log\esssup_{x,y}\frac{\W(y|x)}{Q_Y(y)}
=\log\max_y\frac{\sup_x\W(y|x)}{Q_Y(y)},
\]
where the essential supremum is under $P_X\W$ and, accordingly, $\sup_x$
ranges over $\mathrm{supp}(P_X)$; the right side is minimized over $Q_Y$ by
$Q_Y(y)\propto\sup_x\W(y|x)$, whose value is $\log\sum_y\sup_x\W(y|x)$.

\section{Joint Source--Channel Coding Proofs}\label{app:jscc}

We first record the general list-size ($L$-ary) integral/exponent identity used
by the spectral form of \cref{sec:jscc}; at $L=1$ it is the slope identity
\eqref{eq:recap-slope} of~\cite{elkayampep1}.

\begin{lemma}[Integral representation and exponent identity, general $L$]\label{app:lem-integral-L}
Let $Z\ge0$ have continuous CDF $F$, fix $L\ge1$, and set
$P(R)\triangleq\E{\min\{1,e^{L(R-Z)}\}}$, $E(R)\triangleq-\log P(R)$. Then
\begin{equation}\label{eq:integral-L}
P(R)=L\,e^{LR}\!\int_R^\infty F(z)\,e^{-Lz}\,dz,
\end{equation}
and at every $R$ with $F(R)>0$, $P\in C^1$ and
\begin{equation}\label{eq:expid-L}
P(R)=\frac{F(R)}{1+\tfrac1L\dot E(R)},\qquad -L<\dot E(R)\le0.
\end{equation}
\end{lemma}

\begin{IEEEproof}
Writing $P(R)=\int_0^R dF+\int_R^\infty e^{L(R-z)}dF(z)
=F(R)+e^{LR}\int_R^\infty e^{-Lz}dF(z)$ and integrating the last integral by
parts (using $e^{-Lz}F(z)\to0$) gives \eqref{eq:integral-L}; continuity of $F$
makes $P\in C^1$ wherever $P>0$. Differentiating \eqref{eq:integral-L},
$P'(R)=L\,P(R)-L\,F(R)$, so $\dot E(R)=-P'(R)/P(R)=L\bigl(F(R)/P(R)-1\bigr)$,
which rearranges to \eqref{eq:expid-L}. The integrand of \eqref{eq:integral-L}
satisfies $F(z)\ge F(R)$ for $z\ge R$, giving $P(R)\ge F(R)$ and hence
$\dot E(R)\le0$; and $\dot E(R)>-L$ exactly when $F(R)>0$.
\end{IEEEproof}

\begin{IEEEproof}[Proof of Theorem~\ref{thm:jscc-rc}]
Condition on the transmitted pair $(v,x)$, the output $y$, and the transmitted
dither $u_v$. For a competitor symbol $v'\in\cV$ let $\cE_{v'}$ be the event that
$v'$ outranks $(v,x)$; by the lexicographic equivalence of
Proposition~\ref{thm:pep-uniform}(i),
\[
  \cE_{v'}=\BRAs{\PEPU{v',X(v')}{y}{U_{v'}}<\PEPU{v,x}{y}{u_v}},
\]
where $X(v')\sim Q_{X|V}(\cdot\mid v')$ and $U_{v'}\sim\uU$. These events are
conditionally independent across distinct $v'$, each depending only on the
independent draw $(X(v'),U_{v'})$. An $L$-list error is the event that at least
$L$ competitors outrank the transmitted symbol, i.e.\ the union over $L$-subsets
$S$ of $\bigcap_{v'\in S}\cE_{v'}$. The union bound and conditional independence
give
\[
  \PR{\cE^L_{v,x,y,u_v}}
  \le\sum_{\substack{S\subset\cV\setminus\{v\}\\|S|=L}}\prod_{v'\in S}\PR{\cE_{v'}}
  \le\sum_{\substack{S\subset\cV\\|S|=L}}\prod_{v'\in S}\PR{\cE_{v'}}.
\]
By Proposition~\ref{app:prop-bp} (Schur-concavity of the elementary symmetric
polynomial) the right side is at most
$\binom{|\cV|}{L}\bar p^{\,L}$, where
$\bar p=\tfrac1{|\cV|}\sum_{v'\in\cV}\PR{\cE_{v'}}$ is the average competitor
probability. Enlarging the competitor index set to all of $\cV$ leaves this
average equal to the PEP: with $\bar V\sim\Unif{\cV}$,
\[
  \bar p=\PR{\PEPU{\bar V,X(\bar V)}{y}{U}<\PEPU{v,x}{y}{u_v}}=\PEPU{v,x}{y}{u_v},
\]
the last step being the uniformity of the competitor PEP
(Proposition~\ref{thm:pep-uniform}(iii)) over the uniform auxiliary $\bar V$. Hence
$\PR{\cE^L_{v,x,y,u_v}}\le\binom{|\cV|}{L}\PEPU{v,x}{y}{u_v}^L$. Randomizing the
transmitted dither $u_v\to U$, clipping the probability at $1$, and taking the
expectation over $(V,X,Y)\sim P_V\cdot Q_{X|V}\cdot\W$ gives~\eqref{eq:jscc-rc-bin};
the bound $\binom{n}{k}\le(en/k)^k$ then gives~\eqref{eq:jscc-rc-ub}.
\end{IEEEproof}

\begin{proposition}\label{app:prop-bp}
If $p_1,\dots,p_n\in[0,1]$ with mean $\bar p$, then
$\sum_{|S|=k}\prod_{i\in S}p_i\le\binom{n}{k}\bar p^{\,k}$.
\end{proposition}

\begin{IEEEproof}
The elementary symmetric polynomial $e_k(\vec p)$ is Schur-concave: replacing any
$p_1>p_2$ by their average $\tfrac{p_1+p_2}2$ increases $e_k$, since
$e_k(\vec p)=e_k(\vec p')+(p_1+p_2)e_{k-1}(\vec p')+p_1p_2\,e_{k-2}(\vec p')$
(with $\vec p'$ the remaining coordinates) and
$(\tfrac{p_1+p_2}2)^2\ge p_1p_2$. The fixed-sum maximizer is therefore the
constant vector, where $e_k=\binom{n}{k}\bar p^{\,k}$.
\end{IEEEproof}

\begin{IEEEproof}[Proof of Theorem~\ref{thm:jscc-converse}]
Fix $y$ and work over the candidate alphabet $\cJ=\cV$ with scores
$d(v)=m(v,x_v,y)$ and i.i.d.\ dithers. The per-index form
\eqref{eq:rank-perindex} of Lemma~\ref{app:lem-rank} with $M=|\cV|$ states that
for each \emph{fixed} $v$, the PEP event and the rank event have the same
probability over the dithers alone. This is the point that licenses the
averaging step: conditioned on $Y=y$, the transmitted symbol $V$ follows its
(non-uniform) posterior, which is independent of the dithers, so the per-index
identity may be averaged under that posterior --- only the \emph{competitor}
draw inside $\PEPs{\cdot}$ need be uniform, and it is, by construction
($Q_V=\Unif{\cV}$). A list error is $\mathrm{rank}(V)>L$, so
$\PRs{}{\text{error}\mid Y=y}=\PR{\PEPs{V}>L/|\cV|}$. Substituting the scores
into the definition of $\PEPs{\cdot}$ recovers the dithered JSCC PEP verbatim,
$\PEPs{v}=\PEPU{v,x_v}{y}{U_v}$, with competitor prior
$Q_V\cdot Q^{(\cC)}_{X|V}$ (the point-mass prior in the competitor role); since
$X=x_V$, $\PEPs{V}=\PEP{V,X}{y}$. Averaging over $Y$ and promoting $>$ to $\ge$
by atomlessness (Proposition~\ref{lem:atomless}; the hypothesis $P_V\ll Q_V$ holds
trivially, $Q_V$ having full support) yields
$P^L_e(\cC)=\PR{\PEP{V,X}{Y}\ge L/|\cV|}$. The fixed code realizes the prior
$Q^{(\cC)}_{X|V}$ in both roles, and the threshold $L/|\cV|$ corresponds to rate
$\log|\cV|-\log L$, giving the infimum bound.
\end{IEEEproof}

\section{Erasure and Correct-Decoding Proofs}\label{app:era}

Throughout, $(X,Y)\sim Q_X\cdot\W$ with $U\sim\uU$ independent. The proof of
\cref{thm:era-ach} uses the default size convention $M=\lceil e^{R}\rceil$;
the affine-identity derivation uses the converse-natural $M=e^{R}$; the
correct-decoding proofs use $M-1=e^{R}$ (the achievability-natural
convention, \cref{rem:size-conv}).

\begin{IEEEproof}[Proof of Theorem~\ref{thm:era-ach}]
Average over the i.i.d.\ $Q_X$ ensemble; by symmetry condition on $I=1$. Let
$Z\triangleq-\log\PEPU{X_1}{Y}{U_1}$ be the transmitted metric and
$C_i\triangleq-\log\PEPU{X_i}{Y}{U_i}$, $i\ge 2$, the competitors. By
\cref{thm:pep-uniform}, given $(X_1,Y)$ the $\{C_i\}_{i\ge2}$ are i.i.d.\ with
$\PR{C_i\ge c}=e^{-c}$, and by \cref{lem:atomless} their law is atomless, so the
strict/non-strict form of clause~(2) is immaterial.

\emph{Erasure.} An erasure occurs if the transmitted codeword fails the absolute
threshold ($Z<R+\gamma$) or some competitor lies within the margin
($C_i>Z-\gamma$). Hence
\begin{align*}
P_{\mathrm{era}}
&\le\PR{Z\le R+\gamma}
+\!\int_{R+\gamma}^\infty\!\PR{\exists i\ge2:C_i\ge z-\gamma}\,d\Fspec{Q_X}{z}\\
&\overset{(a)}{\le}\Fspec{Q_X}{R+\gamma}
+\!\int_{R+\gamma}^\infty\!(M-1)\,e^{-(z-\gamma)}\,d\Fspec{Q_X}{z}\\
&\overset{(b)}{\le}\Fspec{Q_X}{R+\gamma}
+e^{R+\gamma}\!\int_{R+\gamma}^\infty\!e^{-z}\,d\Fspec{Q_X}{z}
=\tilde P_e(R+\gamma;Q_X),
\end{align*}
where $(a)$ is the union bound over the $M-1$ competitors
($\PR{C_i\ge z-\gamma}=e^{-(z-\gamma)}$) and $(b)$ uses $M-1\le e^{R}$; the last
equality is the integration-by-parts identity of \eqref{eq:era-rcu} applied at
threshold $R+\gamma$.

\emph{Undetected error.} A wrong commitment requires some competitor to clear
$R+\gamma$ \emph{and} beat the transmitted metric by $\gamma$. Conditioning on
$Z$,
\begin{align*}
P_{\mathrm{ue}}
&\le e^{-\gamma}\,\PR{Z\le R}
+\!\int_R^\infty\!\PR{\exists i\ge2:C_i\ge z+\gamma}\,d\Fspec{Q_X}{z}\\
&\overset{(c)}{\le}e^{-\gamma}\,\Fspec{Q_X}{R}
+\!\int_R^\infty\!(M-1)\,e^{-(z+\gamma)}\,d\Fspec{Q_X}{z}\\
&\le e^{-\gamma}\!\left[\Fspec{Q_X}{R}
+e^{R}\!\int_R^\infty\!e^{-z}\,d\Fspec{Q_X}{z}\right]
=e^{-\gamma}\,\tilde P_e(R;Q_X).
\end{align*}
In $(c)$, when $Z\le R$ a commitment needs a competitor above $R+\gamma$, of
probability $\le e^{R}e^{-(R+\gamma)}=e^{-\gamma}$; when $Z=z>R$ it needs a
competitor above $z+\gamma$, of probability $\le e^{R-(z+\gamma)}$ (the split
point is $R$, not $R+\gamma$, because an undetected error needs a competitor
clearing the \emph{larger} of the confidence bar $R+\gamma$ and the margin bar
$z+\gamma$). Integration by parts as in \eqref{eq:era-rcu} closes the bound.
\end{IEEEproof}

\begin{IEEEproof}[Derivation of the empirical-prior affine identity (\cref{subsec:era-empirical})]
Under the empirical prior the PEP lives on the $1/M$ grid plus dither. If the
transmitted codeword is strictly outranked by some competitor, then
$\PEPU{X}{Y}{U}\ge G_{X,Y}\ge1/M$, so with $R=\log M$ the threshold events
$\{\PEPU{X}{Y}{U}\ge e^{-R}\}$ and $\{\PEPU{X}{Y}{U}\ge e^{-(R+\gamma)}\}$ are
both certain and the affine relation holds trivially as
$1=e^{-\gamma}\cdot1+(1-e^{-\gamma})$. Otherwise the transmitted codeword lies
in the top tie block, $G_{X,Y}=0$, and its PEP is the pure dither $UH_{X,Y}$
with $H_{X,Y}\ge1/M$; the two tail probabilities are then
$1-e^{-(R+\gamma)}/H_{X,Y}$ and $1-e^{-R}/H_{X,Y}$, which satisfy the same
pointwise relation. Averaging over $(X,Y)$ and invoking
\eqref{eq:recap-converse} gives
$\Fspec{Q_X^{\cC}}{R+\gamma}=1-e^{-\gamma}\bigl(1-P_e(\cC)\bigr)$. For the
commit probability: when the metric scores at each output are distinct, the
top codeword's PEP is the pure dither $U/M$ and clause~(2) of
\cref{def:era-rule} is implied (every competitor's PEP being at least $1/M$),
so the rule commits iff $U\le e^{-\gamma}$, giving
$P_{\mathrm{era}}(\cC)=1-e^{-\gamma}$ and
$P_{\mathrm{ue}}(\cC)=e^{-\gamma}P_e(\cC)$; with ties, the commit probability
depends on the data only through tie-block sizes.
\end{IEEEproof}

\begin{IEEEproof}[Proof of Theorem~\ref{thm:era-cd-ach}, integral form \eqref{eq:era-pc-integral}]
By the exact random-coding identity \eqref{eq:era-pc-exact},
$\bar P_c(R;Q_X)=\E{g(Z)}$ with $Z\triangleq-\log\PEPU{X}{Y}{U}\in[0,\infty]$
and the convention $g(\infty)=1$. Since $g(0)=0$ and $g$ is $C^1$ on
$(0,\infty)$ with $g'\ge0$ and $\int_0^\infty g'(z)\,dz=1$, the layer-cake
representation
\[
g(Z)=\int_0^\infty g'(z)\,\Ind{z<Z}\,dz
\]
holds for every value of $Z$, including $Z=\infty$. Taking expectations and
exchanging expectation and integral (the integrand is nonnegative, so Tonelli's
theorem applies),
\[
\bar P_c(R;Q_X)=\int_0^\infty g'(z)\,\PR{Z>z}\,dz.
\]
Finally, $\PR{Z>z}$ and $F_c(Q_X;z)=\PR{Z\ge z}$ differ only at the (at most
countably many) atoms of $Z$, a Lebesgue-null set, so replacing $\PR{Z>z}$ by
$F_c(Q_X;z)$ leaves the integral unchanged, which is \eqref{eq:era-pc-integral}.
(The argument deliberately avoids integration by parts against $dF_c$, which
requires care when the law of $Z$ has atoms; the layer-cake form does not.)
\end{IEEEproof}

\begin{IEEEproof}[Proof of Theorem~\ref{thm:era-cd-ach}, localization \eqref{eq:era-pc-localized}]
Three facts drive the proof: (i) $g$ is non-decreasing with $g(0)=0$,
$g(\infty)=1$, so $g'$ is a probability density on $[0,\infty)$; (ii) $F_c$ is
non-increasing with $F_c(Q_X;z)=1$ for $z\le0$ (since $Z\ge0$ a.s., which
handles the branch $R<\gamma$ below); (iii) the arithmetic identity
$(M-1)e^{-(R\pm\gamma)}=e^{\mp\gamma}$ under $M-1=e^{R}$.

\emph{Upper bound.} Split the integral \eqref{eq:era-pc-integral} at $R-\gamma$:
on the small-$z$ piece use $F_c\le1$, on the large-$z$ piece use
$F_c(z)\le F_c(R-\gamma)$, giving
\[
\bar P_c(R;Q_X)\le g(R-\gamma)+F_c(Q_X;R-\gamma)\bigl(1-g(R-\gamma)\bigr)
\le g(R-\gamma)+F_c(Q_X;R-\gamma).
\]
For $R\ge\gamma$, the bound $(1-u)^t\le e^{-tu}$ with $t=M-1=e^{R}$ and
$u=e^{-(R-\gamma)}\in[0,1]$ gives
$g(R-\gamma)=\bigl(1-e^{-(R-\gamma)}\bigr)^{M-1}\le
\exp\bigl[-(M-1)e^{-(R-\gamma)}\bigr]=e^{-e^{\gamma}}$; for $R<\gamma$ the
claimed bound is trivial since then $F_c(Q_X;R-\gamma)=1\ge\bar P_c$.

\emph{Lower bound.} Drop the tail past $R+\gamma$ and use
$F_c(z)\ge F_c(R+\gamma)$ on $[0,R+\gamma]$:
$\bar P_c(R;Q_X)\ge F_c(Q_X;R+\gamma)\,g(R+\gamma)$. Bernoulli's inequality,
$(1-u)^t\ge1-tu$ for $u\in[0,1]$ and real $t\ge1$, with $t=M-1=e^{R}$ and
$u=e^{-(R+\gamma)}$, gives
$g(R+\gamma)\ge1-(M-1)e^{-(R+\gamma)}=1-e^{-\gamma}$, whence the lower bound in
\eqref{eq:era-pc-localized}. The tuned form \eqref{eq:era-pc-tuned} is the case
$\gamma=\log R$.
\end{IEEEproof}

\begin{IEEEproof}[Proof of Theorem~\ref{thm:era-cd-meta}]
This is the upper-tail instance of the reverse-channel identity
(\cref{thm:reverse-channel}). Using the rank-equivalent metric
$m(x,y)=\W(y\mid x)/P_Y(y)$, the supremum form \eqref{eq:rc-sup} evaluated at
\[
F_c(Q_X;z)=\PRs{Q_X\cdot\W}{\PEP{X}{Y}\le e^{-z}}
=\Es{Q_X\times P_Y}{m\,\Ind{\PEP{X}{Y}\le e^{-z}}}
\]
gives $e^{-z}\sup_{W^*}\Es{P_Y\cdot W^*}{\W(Y\mid X)/P_Y(Y)}$, and the $P_Y$
factor cancels against the $P_Y\cdot W^*$ measure to yield
\eqref{eq:era-cd-meta}; the metric is finite exactly on $\mathrm{supp}(P_Y)$,
whence the support restriction, and $\Es{Q_X\times P_Y}{m}\le1<\infty$
discharges the integrability hypothesis of \cref{thm:reverse-channel}.
\end{IEEEproof}

\section{Multiple-Access Proofs}\label{app:mac}

\begin{IEEEproof}[Proof of Theorem~\ref{thm:mac-rc}]
\emph{Step 1: conditioning and union bound.} By symmetry assume the transmitted
pair is $(V_1,V_2)=(v_1,v_2)$ with codewords $(X_1(v_1),X_2(v_2))=(x_1,x_2)$.
Condition on the transmitted symbols, codewords, channel output $y$, and dither
$U_{v_1,v_2}$. An error requires a competing pair $(v_1',v_2')$ to beat
$(v_1,v_2)$ in the lexicographic order; the competitors split into the three
disjoint classes of \cref{def:mac-pep} (first coordinate only, second
coordinate only, both). Because one independent dither is attached per candidate
pair, the competing dithers appearing below are mutually independent of
$U_{v_1,v_2}$ and of one another. The conditional error probability is at most
$A+B+C$, the union-bound sums over the three classes.

\emph{Step 2: the row class.} Bound $A$ by the four-step chain:
\begin{itemize}
\item[(a)] extend the sum from $v_1'\ne v_1$ to all $v_1'\in\cV_1$ (each added
term is nonnegative);
\item[(b)] apply the lexicographic equivalence (\cref{thm:pep-uniform}(i)) on
the row candidate space: each beat event becomes
$\{\PEPU{v_1',X_1(v_1')}{v_2,x_2,y}{U_{v_1',v_2}}<
\PEPU{v_1,x_1}{v_2,x_2,y}{U_{v_1,v_2}}\}$;
\item[(c)] average the competitor under the uniform reference
$Q_{V_1}=\Unif{\cV_1}$, converting the sum over $v_1'$ into
$|\cV_1|$ times a probability over $(\bar V_1,\bar X_1)\sim
Q_{V_1}\cdot Q_{X_1\mid V_1}$;
\item[(d)] apply uniformity (\cref{thm:pep-uniform}(iii)): the randomized row
PEP of the competitor is $\uU$-distributed, so it beats the transmitted PEP with
probability equal to the latter's value.
\end{itemize}
The chain gives $A\le|\cV_1|\,\PEPU{v_1,x_1}{v_2,x_2,y}{U_{v_1,v_2}}$;
averaging over the transmitted variables, $\E{A}\le E_1$.

\emph{Step 3: the column class.} Symmetric, with the roles of the terminals
swapped: $\E{B}\le E_2$.

\emph{Step 4: the joint class.} Apply the \emph{same} chain (a)--(d) on the
product candidate space $\cV_1\times\cV_2$: (a) extend the sum from the
$(|\cV_1|-1)(|\cV_2|-1)$ joint-class pairs to all of $\cV_1\times\cV_2$ (the
per-pair dithers $U_{v_1',v_2'}$ are independent of $U_{v_1,v_2}$); (b) convert
each beat event into a comparison of joint PEPs by lexicographic equivalence;
(c) average the competitor pair under the uniform product reference
$Q_{V_1}\times Q_{V_2}$, yielding the factor $|\cV_1||\cV_2|$; (d) apply
uniformity of the randomized joint PEP. The individual beat probabilities are
\emph{not} equal across joint-class competitors --- the metric and the
conditional priors depend on $(v_1',v_2')$ --- but the chain never needs them to
be. Hence $\E{C}\le E_{1,2}$.

\emph{Step 5: conclusion.} The conditional error probability is at most
$\min\{1,A+B+C\}$; taking expectations and using
$\E{\min\{1,A+B+C\}}\le\E{\min\{1,E_1+E_2+E_{1,2}\}}$ after the substitutions of
Steps 2--4 yields \eqref{eq:mac-rc}.
\end{IEEEproof}

\begin{IEEEproof}[Proof of Theorem~\ref{thm:mac-converse}]
\emph{The exact identity \eqref{eq:mac-conv}.} The joint candidate space
$\cV_1\times\cV_2$ with channel-input space $\cX_1\times\cX_2$ and the product
uniform reference $Q_{V_1}\times Q_{V_2}=\Unif{\cV_1\times\cV_2}$ is a
single-transmitter JSCC problem, with the code inducing the point-mass prior
$Q^{(\cC)}(x_1,x_2\mid v_1,v_2)=\delta_{(x_1(v_1),x_2(v_2))}$. \Cref{thm:jscc-converse} at $L=1$ (its per-index mechanism, Lemma~\ref{app:lem-rank},
averaged under the joint source law $P_{V_1,V_2}$, which need not be uniform)
gives
\[
P_e(\cC)=\PR{\PEP{V_1,X_1,V_2,X_2}{Y}\ge\tfrac1{|\cV_1||\cV_2|}}=\PR{E_{1,2}\ge1}.
\]

\emph{The marginal bounds.} Compare the decoder against a \emph{row-restricted}
benchmark: a genie reveals $(V_2,X_2)$ to the decoder, which then runs the
\emph{same} metric $m$ with the \emph{same} dithers, restricted to the row of
candidates $\{(v_1',V_2):v_1'\in\cV_1\}$. The comparison is pathwise: for every
realization of the source pair, codewords, channel output, and dithers, the
event that the transmitted pair is lexicographically maximal over the whole
candidate grid is contained in the event that it is maximal over its own row.
Hence correct decoding by the full decoder implies correct decoding by the
row-restricted one, and
\[
P_e(\cC)\ \ge\ P_e^{\mathrm{row}}(\cC),
\]
for any metric, any tie realization, and any source dependence. (The familiar
maxim that side information cannot increase the error probability concerns
\emph{optimal} decoders; it does not by itself cover the fixed, possibly
mismatched, metric decoder used here, whose row-restricted benchmark is not the
optimal side-informed decoder unless $m$ is matched MAP. The pathwise inclusion
closes this gap.)

With the row restriction in place, the problem reduces to a point-to-point JSCC
problem for terminal~1: the conditional channel is $W(\cdot\mid\cdot,x_2)$, the
ensemble for terminal~1 is $Q_{V_1}\times Q_{X_1\mid V_1}$ with $(v_2,x_2)$
frozen, and the conditional row PEP $\PEP{v_1,x_1}{v_2,x_2,y}$ is exactly the
JSCC PEP of this reduced problem. Applying \cref{thm:jscc-converse} pointwise
for each realization $(V_2,X_2)=(v_2,x_2)$ and integrating over
$(V_2,X_2)\sim P_{V_2}\cdot Q_{X_2\mid V_2}$ yields
$P_e^{\mathrm{row}}(\cC)=\PR{\PEP{V_1,X_1}{V_2,X_2,Y}\ge1/|\cV_1|}
=\PR{E_1\ge1}$. The second marginal bound follows by symmetry upon revealing
$V_1$ instead, and \eqref{eq:mac-conv} combines the exact identity with the two
marginal bounds.
\end{IEEEproof}

\bibliographystyle{IEEEtran}
\bibliography{refs}

\begin{thebibliography}{10}
\providecommand{\url}[1]{#1}
\csname url@samestyle\endcsname
\providecommand{\newblock}{\relax}
\providecommand{\bibinfo}[2]{#2}
\providecommand{\BIBentrySTDinterwordspacing}{\spaceskip=0pt\relax}
\providecommand{\BIBentryALTinterwordstretchfactor}{4}
\providecommand{\BIBentryALTinterwordspacing}{\spaceskip=\fontdimen2\font plus
\BIBentryALTinterwordstretchfactor\fontdimen3\font minus
  \fontdimen4\font\relax}
\providecommand{\BIBforeignlanguage}[2]{{%
\expandafter\ifx\csname l@#1\endcsname\relax
\typeout{** WARNING: IEEEtran.bst: No hyphenation pattern has been}%
\typeout{** loaded for the language `#1'. Using the pattern for}%
\typeout{** the default language instead.}%
\else
\language=\csname l@#1\endcsname
\fi
#2}}
\providecommand{\BIBdecl}{\relax}
\BIBdecl

\bibitem{elkayampep1}
N.~Elkayam and M.~Feder, ``A pairwise-error-probability framework for one-shot
  information theory,'' \emph{arXiv preprint arXiv:2608.06577}, 2026.

\bibitem{matsuta2015non}
T.~Matsuta and T.~Uyematsu, ``Non-asymptotic bounds for fixed-length lossy
  compression,'' in \emph{2015 IEEE International Symposium on Information
  Theory (ISIT)}.\hskip 1em plus 0.5em minus 0.4em\relax IEEE, 2015, pp.
  1811--1815, journal version: {IEICE} Trans. Fundamentals, vol.~E99-A, no.~12,
  pp.~2116--2129, 2016.

\bibitem{kostina2012fixed}
V.~Kostina and S.~Verd{\'u}, ``Fixed-length lossy compression in the finite
  blocklength regime,'' \emph{IEEE Transactions on Information Theory},
  vol.~58, no.~6, pp. 3309--3338, 2012.

\bibitem{vazquez2016bayesian}
G.~Vazquez-Vilar, A.~T. Campo, A.~G. i~F{\`a}bregas, and A.~Martinez,
  ``Bayesian {M}-ary hypothesis testing: The meta-converse and
  {Verd{\'u}}-{Han} bounds are tight,'' \emph{IEEE Transactions on Information
  Theory}, vol.~62, no.~5, pp. 2324--2333, 2016.

\bibitem{elkayam2020one}
N.~Elkayam and M.~Feder, ``One shot approach to lossy source coding under
  average distortion constraints,'' in \emph{2020 IEEE International Symposium
  on Information Theory (ISIT)}.\hskip 1em plus 0.5em minus 0.4em\relax IEEE,
  2020, pp. 2389--2393.

\bibitem{elkayam2019oneshot_rate_distortion_full}
\BIBentryALTinterwordspacing
------, ``One shot approach to lossy source coding under average distortion
  constraints,'' 2020, arXiv:2001.03983. [Online]. Available:
  \url{https://arxiv.org/abs/2001.03983}
\BIBentrySTDinterwordspacing

\bibitem{palzer2016converse}
L.~Palzer and R.~Timo, ``A converse for lossy source coding in the finite
  blocklength regime,'' in \emph{24th International Zurich Seminar on
  Communications (IZS)}.\hskip 1em plus 0.5em minus 0.4em\relax ETH-Z{\"u}rich,
  2016.

\bibitem{ElkayamLD2017}
\BIBentryALTinterwordspacing
N.~Elkayam and M.~Feder, ``A general approach to list decoding,'' 2017, online.
  [Online]. Available: \url{http://www.eng.tau.ac.il/~elkayam/ListDecoding.pdf}
\BIBentrySTDinterwordspacing

\bibitem{polyanskiy2010channel}
Y.~Polyanskiy, H.~V. Poor, and S.~Verd{\'u}, ``Channel coding rate in the
  finite blocklength regime,'' \emph{IEEE Transactions on Information Theory},
  vol.~56, no.~5, pp. 2307--2359, 2010.

\bibitem{han2003information}
T.~S. Han, \emph{Information-Spectrum Methods in Information Theory}.\hskip 1em
  plus 0.5em minus 0.4em\relax Berlin: Springer, 2003.

\bibitem{campo2011random}
A.~T. Campo, G.~Vazquez-Vilar, A.~G. i~F{\`a}bregas, and A.~Martinez,
  ``Random-coding joint source-channel bounds,'' in \emph{2011 IEEE
  International Symposium on Information Theory Proceedings}.\hskip 1em plus
  0.5em minus 0.4em\relax IEEE, 2011, pp. 899--902.

\bibitem{merhav2014list}
N.~Merhav, ``List decoding—random coding exponents and expurgated
  exponents,'' \emph{IEEE Transactions on Information Theory}, vol.~60, no.~11,
  pp. 6749--6759, 2014.

\bibitem{tan2014second}
V.~Y. Tan and P.~Moulin, ``Second-order capacities of erasure and list
  decoding,'' in \emph{2014 IEEE International Symposium on Information
  Theory}.\hskip 1em plus 0.5em minus 0.4em\relax IEEE, 2014, pp. 1887--1891.

\bibitem{csiszar1980joint}
I.~Csisz{\'a}r, ``Joint source-channel error exponent,'' \emph{Problems of
  Control and Information Theory}, vol.~9, pp. 315--328, 1980.

\bibitem{kostina2013jscc}
V.~Kostina and S.~{Verd\'u}, ``Lossy joint source-channel coding in the finite
  blocklength regime,'' \emph{IEEE Transactions on Information Theory},
  vol.~59, no.~5, pp. 2545--2575, 2013.

\bibitem{forney1968exponential}
G.~D. Forney, Jr., ``Exponential error bounds for erasure, list, and decision
  feedback schemes,'' \emph{IEEE Transactions on Information Theory}, vol.~14,
  no.~2, pp. 206--220, 1968.

\bibitem{HaimKE18}
\BIBentryALTinterwordspacing
E.~Haim, Y.~Kochman, and U.~Erez, ``On random-coding union bounds with and
  without erasures,'' \emph{IEEE Transactions on Information Theory}, vol.~64,
  no.~6, pp. 4294--4308, 2018. [Online]. Available:
  \url{https://doi.org/10.1109/TIT.2018.2825357}
\BIBentrySTDinterwordspacing

\bibitem{shannon1967lower}
C.~E. Shannon, R.~G. Gallager, and E.~R. Berlekamp, ``Lower bounds to error
  probability for coding on discrete memoryless channels. {II},''
  \emph{Information and Control}, vol.~10, no.~5, pp. 522--552, 1967.

\bibitem{arimoto1973converse}
S.~Arimoto, ``On the converse to the coding theorem for discrete memoryless
  channels (corresp.),'' \emph{IEEE Transactions on Information Theory},
  vol.~19, no.~3, pp. 357--359, 1973.

\bibitem{polyanskiy2010arimoto}
Y.~Polyanskiy and S.~Verd{\'u}, ``Arimoto channel coding converse and
  {R{\'e}nyi} divergence,'' in \emph{2010 48th Annual Allerton Conference on
  Communication, Control, and Computing (Allerton)}.\hskip 1em plus 0.5em minus
  0.4em\relax IEEE, 2010, pp. 1327--1333.

\bibitem{slepian1973noiseless}
D.~Slepian and J.~K. Wolf, ``Noiseless coding of correlated information
  sources,'' \emph{IEEE Transactions on Information Theory}, vol.~19, no.~4,
  pp. 471--480, 1973.

\bibitem{chen2020lossless}
S.~Chen, M.~Effros, and V.~Kostina, ``Lossless source coding in the
  point-to-point, multiple access, and random access scenarios,'' \emph{IEEE
  Transactions on Information Theory}, vol.~66, no.~11, pp. 6688--6722, 2020.

\bibitem{ahlswede1971multi}
R.~Ahlswede, ``Multi-way communication channels,'' in \emph{Proc. 2nd Int.
  Symp. Inf. Theory}, Tsahkadsor, Armenian S.S.R., 1971, pp. 23--52.

\bibitem{liao1972multiple}
H.~H.~J. Liao, ``Multiple access channels,'' Ph.D. dissertation, University of
  Hawaii, Honolulu, 1972.

\bibitem{mccormick1976computability}
G.~P. McCormick, ``Computability of global solutions to factorable nonconvex
  programs: Part {I} --- convex underestimating problems,'' \emph{Mathematical
  Programming}, vol.~10, no.~1, pp. 147--175, 1976.

\bibitem{sherali1990hierarchy}
H.~D. Sherali and W.~P. Adams, ``A hierarchy of relaxations between the
  continuous and convex hull representations for zero-one programming
  problems,'' \emph{SIAM J. Discrete Math.}, vol.~3, no.~3, pp. 411--430, 1990.

\bibitem{jose2018improved}
S.~T. Jose and A.~A. Kulkarni, ``Improved finite blocklength converses for
  {Slepian--Wolf} coding via linear programming,'' \emph{IEEE Transactions on
  Information Theory}, vol.~65, no.~4, pp. 2423--2441, 2019.

\bibitem{li2021unified}
C.~T. Li and V.~Anantharam, ``A unified framework for one-shot achievability
  via the {Poisson} matching lemma,'' \emph{IEEE Transactions on Information
  Theory}, vol.~67, no.~5, pp. 2624--2651, 2021.

\bibitem{yassaee2013technique}
M.~H. Yassaee, M.~R. Aref, and A.~Gohari, ``A technique for deriving one-shot
  achievability results in network information theory,'' in \emph{2013 IEEE
  International Symposium on Information Theory}.\hskip 1em plus 0.5em minus
  0.4em\relax IEEE, 2013, pp. 1287--1291.

\bibitem{verdu2012non}
S.~Verd{\'u}, ``Non-asymptotic achievability bounds in multiuser information
  theory,'' in \emph{2012 50th Annual Allerton Conference on Communication,
  Control, and Computing (Allerton)}.\hskip 1em plus 0.5em minus 0.4em\relax
  IEEE, 2012, pp. 1--8.

\bibitem{tan2014dispersions}
V.~Y.~F. Tan and O.~Kosut, ``On the dispersions of three network information
  theory problems,'' \emph{IEEE Transactions on Information Theory}, vol.~60,
  no.~2, pp. 881--903, 2014.

\bibitem{molavianjazi2015second}
E.~MolavianJazi and J.~N. Laneman, ``A second-order achievable rate region for
  {Gaussian} multi-access channels via a central limit theorem for functions,''
  \emph{IEEE Transactions on Information Theory}, vol.~61, no.~12, pp.
  6719--6733, 2015.

\end{thebibliography}

\end{document}